\documentclass[journal,draftcls,onecolumn,12pt,twoside]{IEEEtranTCOM}

\usepackage[english]{babel}
\usepackage{booktabs}

\usepackage{ifpdf}

\usepackage{cite} % Orders citations.
\usepackage{url}
\usepackage{hyperref}
\usepackage{setspace}
\ifCLASSINFOpdf
	\usepackage[pdftex]{graphicx}
	\graphicspath{{./figures/}}
\else
	\usepackage[dvips]{graphicx}
	\graphicspath{./figures/}
\fi
\usepackage{color}

\usepackage{pgf, tikz, pgfplots}
\usetikzlibrary{shapes, arrows, automata}
\usetikzlibrary{calc,hobby,decorations}
\usepackage[cmex10]{amsmath}
\usepackage{amsfonts, amssymb, amsthm}
\usepackage{mathrsfs}

\usepackage{algorithm} %format of the algorithm
\usepackage{algorithmic}
\usepackage{enumerate}
\usepackage{multirow}
\usepackage{rotating}
\usepackage{subcaption}
\usepackage{needspace}

\input{mySymbol.sty}

\definecolor{penndarkestblue}{cmyk}{1,0.74,0,0.77}
\definecolor{penndarkerblue}{cmyk}{1,0.74,0,0.70}
\definecolor{pennblue}{cmyk}{0.99,0.66,0,0.57} 
\definecolor{pennlighterblue}{cmyk}{0.98,0.44,0,0.35}
\definecolor{pennlightestblue}{cmyk}{0.38,0.17,0,0.17} 
\definecolor{penndarkestred}{cmyk}{0,1,0.89,0.66}
\definecolor{penndarkerred}{cmyk}{0,1,0.88,0.55}
\definecolor{pennred}{cmyk}{0,1,0.83,0.42} 
\definecolor{pennlighterred}{cmyk}{0,1,0.6,0.24}
\definecolor{pennlightestred}{cmyk}{0,0.43,0.26,0.12} 
\definecolor{penndarkestgreen}{cmyk}{1,0,1,0.68}
\definecolor{penndarkergreen}{cmyk}{1,0,1,0.57}
\definecolor{penngreen}{cmyk}{1,0,1,0.44} 
\definecolor{pennlightergreen}{cmyk}{1,0,1,0.25}
\definecolor{pennlightestgreen}{cmyk}{0.43,0,0.43,0.13}
\definecolor{penndarkestorange}{cmyk}{0,0.65,1,0.49}
\definecolor{penndarkerorange}{cmyk}{0,0.65,1,0.33}
\definecolor{pennorange}{cmyk}{0,0.54,1,0.24} 
\definecolor{pennlighterorange}{cmyk}{0,0.32,1,0.13}
\definecolor{pennlightestorange}{cmyk}{0,0.15,0.46,0.06}
\definecolor{penndarkestpurple}{cmyk}{0,1,0.11,0.86}
\definecolor{penndarkerpurple}{cmyk}{0,1,0.13,0.82}
\definecolor{pennpurple}{cmyk}{0,1,0.11,0.71} 
\definecolor{pennlighterpurple}{cmyk}{0,1,0.05,0.46}
\definecolor{pennlightestpurple}{cmyk}{0,0.35,0.02,0.23}
\definecolor{pennyellow}{cmyk}{0,0.20,1,0.05} 
\definecolor{pennlightgray1}{cmyk}{0,0,0,0.05}
\definecolor{pennlightgray3}{cmyk}{0.01,0.01,0,0.18}
\definecolor{pennmediumgray1}{cmyk}{0.04,0.03,0,0.31}
\definecolor{pennmediumgray4}{cmyk}{0.08,0.06,0,0.54}
\definecolor{penndarkgray2}{cmyk}{0.09,0.07,0,0.71}
\definecolor{penndarkgray4}{cmyk}{0.1,0.1,0,0.92}
\def\SO3{\mathrm{SO(3)}}

\newtheorem{theorem}{\hspace{0pt}\bf Theorem}

\newtheorem{remark}{\hspace{0pt}\bf Remark}

\newtheorem{definition}{\hspace{0pt}\bf Definition}

\begin{document}
\begin{spacing}{1.57}

\title{Resource Allocation via \\Model-Free Deep Learning in \\Free Space Optical Communications
}
% Authors: full names plus addresses.
\author{Zhan Gao$^\star$, Mark Eisen$^\dagger$, and Alejandro Ribeiro$^\star$ 
\thanks{Preliminary results appear in GLOBECOM 2019 \cite{gao2019optimal}. $^\star$Department of Electrical and Systems Engineering, University of Pennsylvania, USA (Email: \{gaozhan,aribeiro\}@seas.upenn.edu). $^\dagger$Intel Corporation, USA (Email: mark.eisen@intel.com).}

}

\maketitle

\begin{abstract}
This paper investigates the general problem of resource allocation for mitigating channel fading effects in Free Space Optical (FSO) communications. The resource allocation problem is modeled as the constrained stochastic optimization framework, which covers a variety of FSO scenarios involving power adaptation, relay selection and their joint allocation. Under this framework, we propose two algorithms that solve FSO resource allocation problems. We first present the Stochastic Dual Gradient (SDG) algorithm that is shown to solve the problem exactly by exploiting the strong duality but whose implementation necessarily requires explicit and accurate system models. As an alternative we present the Primal-Dual Deep Learning (PDDL) algorithm based on the SDG algorithm, which parametrizes the resource allocation policy with Deep Neural Networks (DNNs) and optimizes the latter via a primal-dual method. The parametrized resource allocation problem incurs only a small loss of optimality due to the strong representational power of DNNs, and can be moreover implemented without knowledge of system models. A wide set of numerical experiments are performed to corroborate the proposed algorithms in FSO resource allocation problems. We demonstrate their superior performance and computational efficiency compared to the baseline methods in both continuous power allocation and binary relay selection settings.

%allow these parameters to evolve. The theoretically optimal step-size is selected, along with the adaptive batch-size which is increased when the convergence error saturates, ensuring provably fast convergence to ever-tighter error neighborhoods that attenuate to null in the limit. 
%

\end{abstract}


\IEEEpeerreviewmaketitle

%!TEX root = mainOp.tex
%%%%%%%%%%%%%%%%%%%%%%%%%%%%%%%
%%% SECTION : Introduction  %%%
%%%%%%%%%%%%%%%%%%%%%%%%%%%%%%%

\section{Introduction} \label{sec:intro}

Free Space Optical (FSO) communication has attracted noticeable attention due to high capacity, low cost, strong security and flexible construction \cite{khalighi2014survey}. It transmits signals with optical carriers through the atmosphere and has found applications in satellite communications \cite{martini2002free}, last-mile access \cite{akella2007multi}, and fronthaul or backhaul for wireless cellular networks \cite{alzenad2018fso}. Despite this potential, FSO communication is susceptible to channel characteristics, such as atmospheric turbulence, weather conditions and background radiation \cite{andrews2005laser}. Different models were proposed to characterize the FSO channel, based on which a number of techniques were developed to mitigate channel effects \cite{borah2009pointing, gao2017beam, gao2018beam, gao2019beam, zhang2019ergodicity, kiasaleh2005performance, navidpour2007ber}. Cooperative transmission has recently been introduced as one of such techniques in FSO communications, which improves the system performance by leveraging optimal resource allocation \cite{abou2010cooperative}. That is, it allocates resources adaptively based on the channel state information (CSI) in order to optimize the system performance. Common examples of FSO resource allocation problems include power adaptation, relay selection and their joint allocation. 

Power adaptation has emerged as a popular cooperative transmission technique to mitigate channel fading effects, but the conventional adaptation method in radio frequency (RF) channels does not apply directly to optical channels \cite{park2013power}. The works in \cite{zhou2015optical, sun2018beam, park2013power} assign adaptive powers to orthogonal optical carriers maximizing the channel capacity with total and peak power constraints. The authors in \cite{abou2011cooperative, abd2017effect} minimize the outage probability with respective power allocation strategies. Other applications include the security performance \cite{abd2017effect}, the spectral efficiency \cite{hassan2018delay}, etc. Relay-assisted communication, on the other hand, employs multiple relay nodes between the transmitter and the receiver to create a virtual multiple-aperture FSO system \cite{safari2008relay, karimi2011free, karimi2009ber1}. However, it is not practical to activate all available relays that requires perfect transmission synchronization. The works in \cite{chatzidiamantis2013relay, abou2013performance} developed relay selection protocols for optimal outage and error probabilities, and the authors in \cite{kashani2013optimal} considered both serial and parallel relays to improve the system performance. Furthermore, joint power and relay allocation algorithms were developed in FSO networks, in order to maximize the network throughout and minimize the outage probability \cite{zhou2013joint, hassan2017statistical, hassan2017delay}. However, the aforementioned works are restricted by both or one of the following limitations.
\begin{enumerate}[]
\item \textbf{L.1} Approximation approaches are required to simplify optimization problems or to obtain convex relaxations.

\item \textbf{L.2} The implementation of these methods requires complete knowledge of system models (e.g., capacity functions and channel distributions).
\end{enumerate}
\textbf{L.1} results in inexact solutions that degrade performance and/or require high computational cost. \textbf{L.2} yields solutions that depend on the system model information, which may be unavailable or inaccurate, thus introducing inevitable errors. \textbf{L.1} and \textbf{L.2} further tie the solution methods to specific use-cases and do not necessarily generalize to changes in system structure, i.e., the methods become inapplicable or require significant modifications when changing FSO systems. These limitations provide an incomplete solution to the design of generic resource allocation policies in FSO communications.

These challenges of existing FSO resource allocation methods make the application of machine learning methods appealing, due to their low complexity, potential for model-free implementation, and transference to unseen scenarios. Deep Neural Networks (DNNs) have been developed as predominant tools to analyze data for target information and have achieved resounding success in many communication, signal processing and control problems \cite{liu2017survey, canziani2016analysis, sanchez2018real}. In particular, DNNs have been applied for power allocation of the interference management problem in wireless RF domain \cite{sun2017learning, xu2017deep, eisen2019learning}. In FSO communication domain, DNNs have been utilized for assisting channel estimation \cite{amirabadi2020deep, lohani2018turbulence} and signal modulation/demodulation \cite{darwesh2020deep, lee2019deep}. While to the best of our knowledge, deep learning approaches have not yet been systematically explored for general resource allocation problems in FSO systems. 

In this paper we study the application of dual domain optimization and deep learning methods in a wide array of resource allocation problems in FSO communications. Given the objective with a set of constraints, we formulate the FSO resource allocation problem as the constrained stochastic optimization problem and seek an optimal resource allocation policy that adapts to the channel state information (Section \ref{sec:problem}). To demonstrate the generality of our framework, we exemplify with problems of power adaptation in Radio on FSO systems (Section \ref{poweradaption}), relay selection in relay-assisted FSO networks (Section \ref{relayselection}), and joint power and relay allocation in FSO fronthaul networks (Section \ref{powerrelay}). Such resource allocation problems are typically challenging due to the non-convexity of complicated objective, existence of constraints, infinite dimensionality of resource allocation policy, and lack of system model knowledge. We propose the use of dual optimization and learning framework to address these challenges and provide a comprehensive solution methodology. More in detail, our contributions are as follows. 

\smallskip
\begin{enumerate}[(i)]
\item We propose the Stochastic Dual Gradient (SDG) algorithm to overcome the limitation \textbf{L.1} (Section \ref{sec_sdg}). The SDG algorithm is demonstrated to solve FSO resource allocation problems exactly by utilizing the strong duality. The latter allows us to operate in the dual domain, which is convex, unconstrained and finite dimensional, without loss of optimality. The SDG further saves computational cost by performing primal-dual gradient updates, which avoids computing KKT conditions. Despite the theoretical advantages, this algorithm is limited as it is model-based that requires knowledge of specific system models.

\item We propose the Primal-Dual Deep Learning (PDDL) algorithm as a model-free, deep learning based alternative to overcome the limitation \textbf{L.2} (Section \ref{sec_pddl}). The PDDL parameterizes the resource allocation policy with DNNs and reformulates the problem as a constrained machine learning problem. It leverages an approximate strong duality to train DNNs with an unsupervised primal-dual method. A model-free implementation is obtained by using the policy gradient method, which does not require the knowledge of system models. The PDDL further achieves lower complexity due to the computational efficiency of DNNs. 

%DNNs parameterize the resource allocation policy when we reformulate the problem as a constrained machine learning problem. We moreover leverage an approximate strong duality result that permits to train DNNs with an unsupervised primal-dual method. A model-free implementation is obtained by using the policy gradient method, which does not require the knowledge of system models. The PDDL further achieves lower complexity due to the computational efficiency of DNNs.

\item The overall methodology resulting from both algorithms does not depend on specific systems or problem settings, and thus is applicable comprehensively in the context of FSO communications. We perform numerical experiments in a variety of practical FSO communication scenarios, including power adaptation, relay selection and their joint allocation (Section \ref{sec:sims}). In all scenarios, we illustrate success of the proposed algorithms, validating their transference to changes in system structure. 
\end{enumerate}

%!TEX root = mainOp.tex
%%%%%%%%%%%%%%%%%%%%%%%%%%%%%%%%%%%%%%%%%%%%%
%%% SECTION : Problem Formulation and Background  %%%
%%%%%%%%%%%%%%%%%%%%%%%%%%%%%%%%%%%%%%%%%%%%%

\section{Problem Formulation} \label{sec:problem}
%\subsection{Problem Formulation}

Consider a general Free Space Optical (FSO) communication system under some form of resource constraints. By adaptively allocating resources using a policy that responds to instantaneous fading effects of the atmospheric channel, we can mitigate these fading effects and optimize the system performance. Denote by $\bbh \in \mathbb{R}^m$ the collected channel state information (CSI) and $\bbr(\bbh) \in \mathbb{R}^n$ a policy that determines the allocated resources based on the observed $\bbh$. The objective function $f(\bbh, \bbr(\bbh))$ measures the system performance that is instantiated on $\bbh$ and $\bbr(\bbh)$. Furthermore, a total of $S$ constraints are imposed either on the resources $\bbr(\bbh)$ or on the objective function $f(\bbh, \bbr(\bbh))$, each of which is represented by a constraint function $c_s(\bbr(\bbh), f(\bbh, \bbr(\bbh)))~\forall~s=1,\ldots,S$. The atmospheric channel is typically considered as a fading process with channel coherence time on the order of milliseconds, such that we shall assume $\bbh$ is drawn from an ergodic and i.i.d block fading process. In this context, the instantaneous system performance tends to vary fast and the long term average performance $\mathbb{E}_\bbh[f(\bbh, \bbr(\bbh))]$ is the more meaningful metric to consider when designing an optimal resource allocation policy. We similarly consider constraints to be satisfied in expectation.

\begin{figure*}%
\centering
\begin{subfigure}{0.33\columnwidth}
\includegraphics[width=1.0\linewidth, height = 0.55\linewidth]{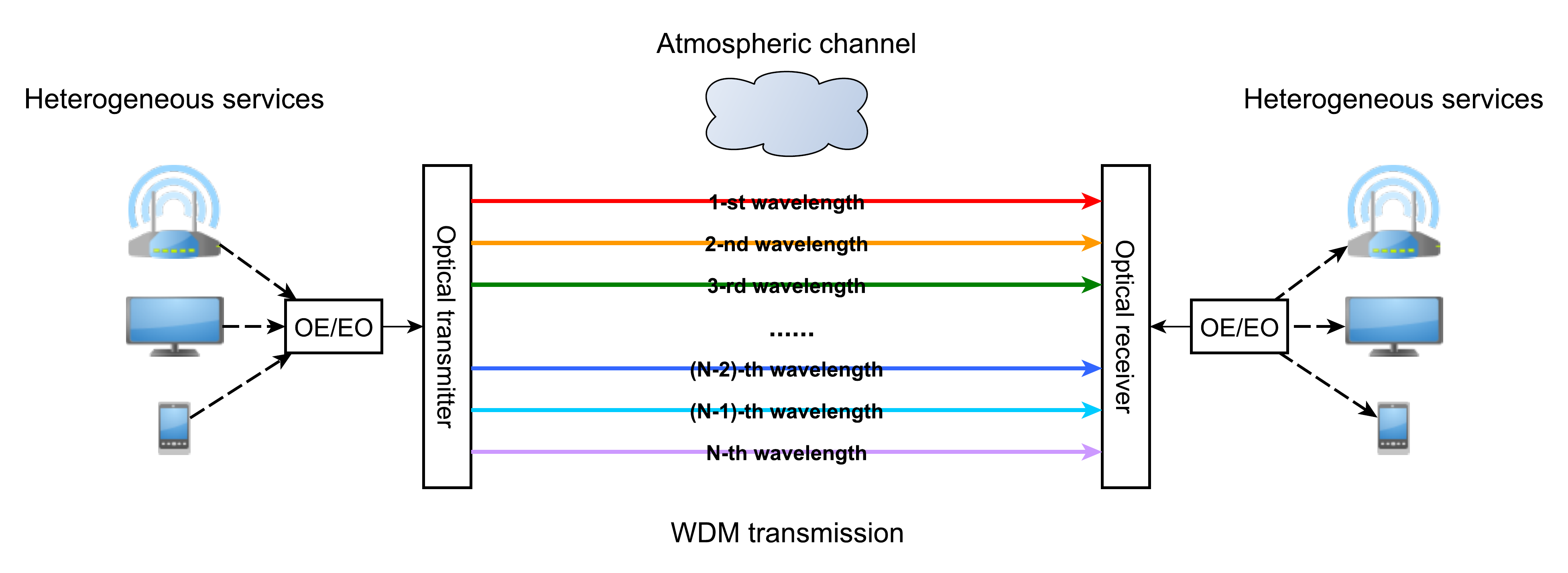}%
\caption{}%
\label{fig_rofso}%
\end{subfigure}\hfill\hfill%
\begin{subfigure}{0.33\columnwidth}
\includegraphics[width=1.0\linewidth,height = 0.55\linewidth]{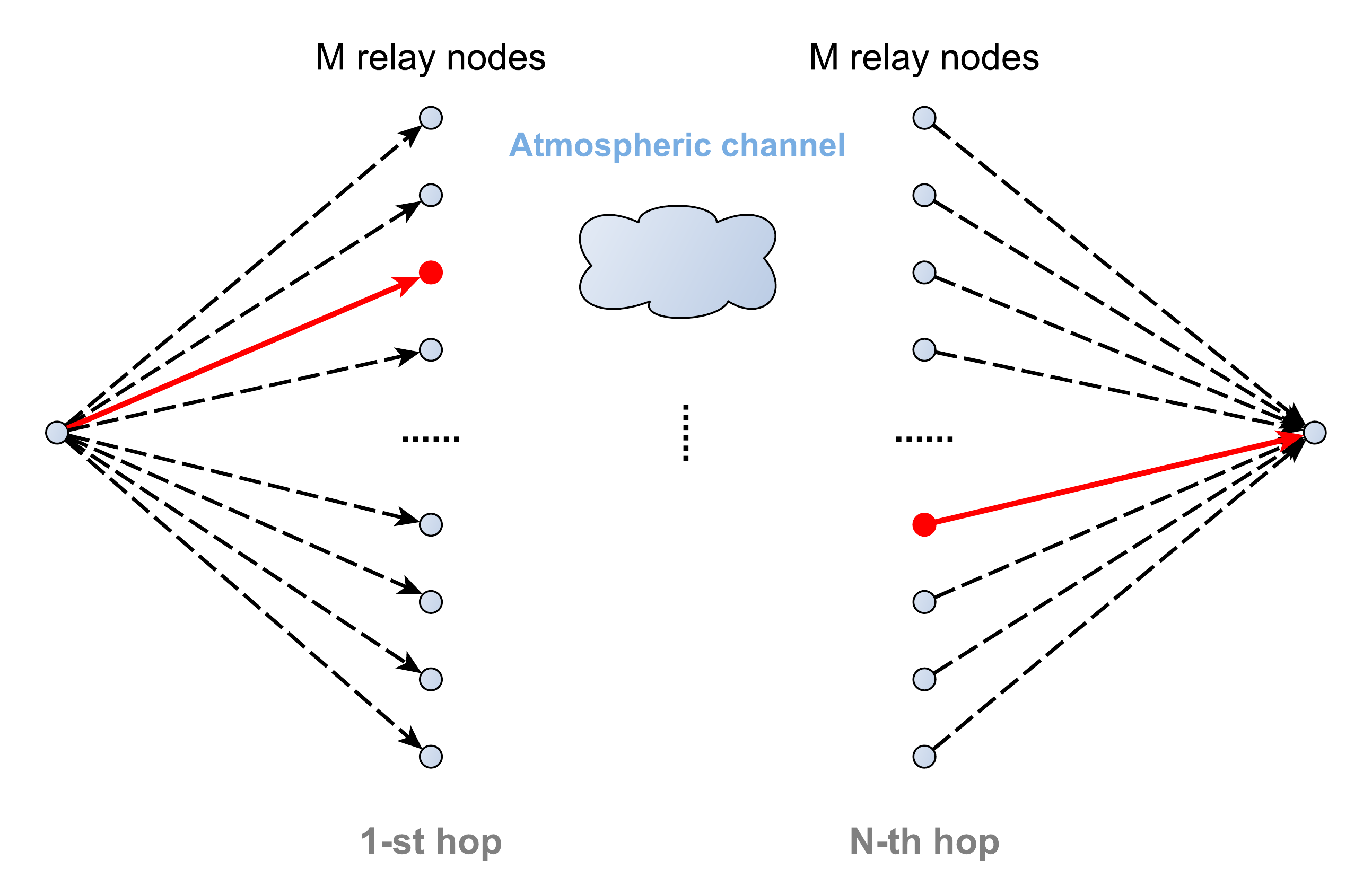}%
\caption{}%
\label{fig_relay1}%
\end{subfigure}\hfill\hfill%
\begin{subfigure}{0.33\columnwidth}
\includegraphics[width=1.0\linewidth,height = 0.55\linewidth]{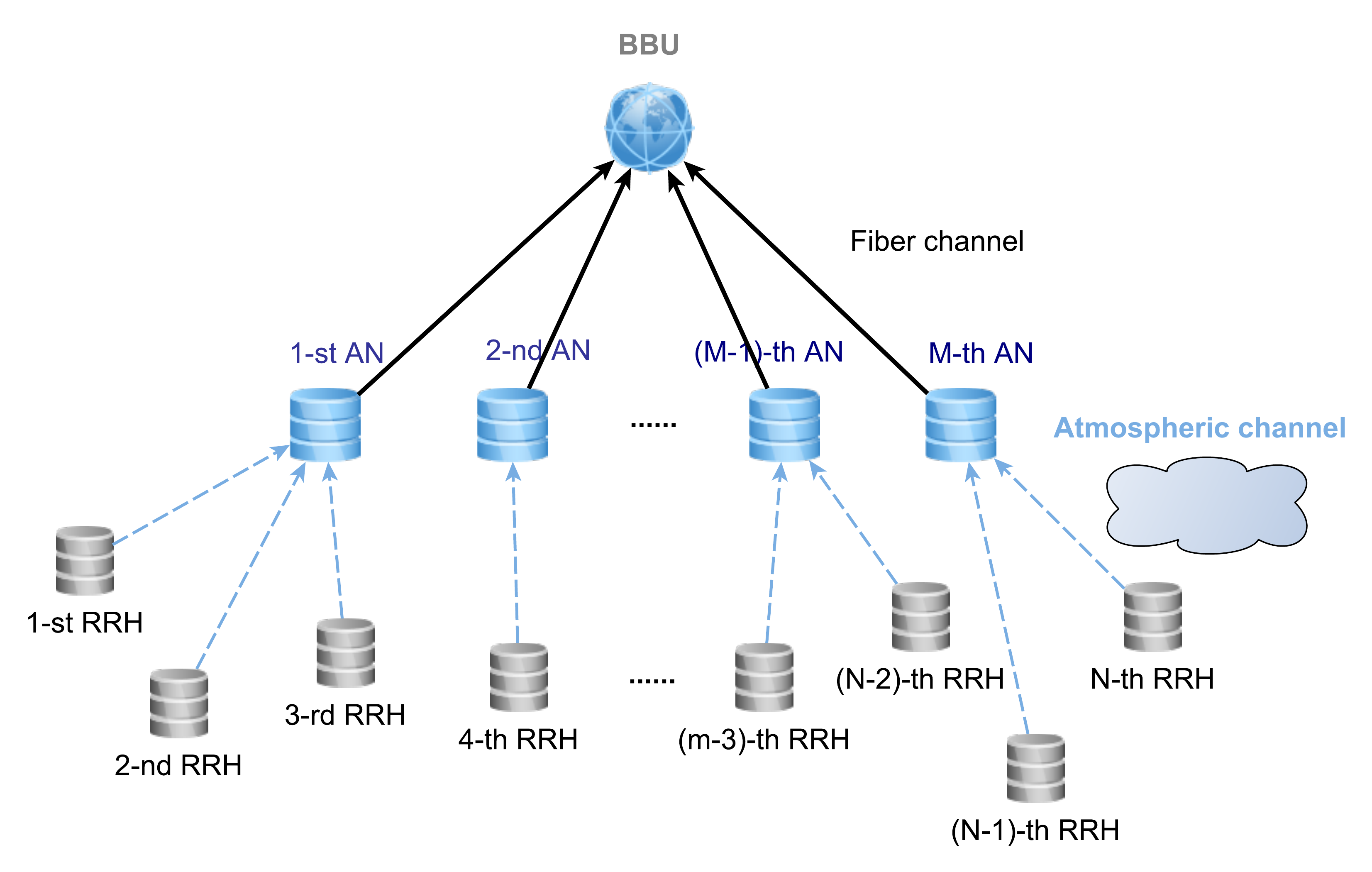}%
\caption{}%
\label{fronthaul}%
\end{subfigure}%
\caption{(a) The DWDM RoFSO system with $N$ optical wavelength channels. (b) The relay-assisted FSO network. The transmitter communicates with the receiver through $N$ selected relays (red nodes). (c) The fronthaul FSO network with $N$ RRHs and $M$ ANs.}\label{FSOsystems}\vspace{-5mm}
\end{figure*}

Our goal is to maximize the expected performance $\mathbb{E}_\bbh[f(\bbh, \bbr(\bbh))]$ given certain resource constraints. In particular, we seek to compute the instantaneous allocated resources $\bbr(\bbh)$ based on the instantaneous CSI $\bbh$, that satisfy required constraints and optimize the system performance. By introducing $\ccalR$ as the action space of allocated resources $\bbr(\bbh)$, we formulate the optimal resource allocation as the following stochastic optimization problem
\begin{alignat}{3} \label{eq_problem_resource}
 \mathbb{P}:= &  \max_{\bbr(\bbh)} \ && \mathbb{E}_\bbh \left[f(\bbh, \bbr(\bbh))\right] ,             \\
        &  \st                    \ && \mathbb{E}_\bbh \! \left[ c_s\big(\bbr(\bbh), f(\bbh, \bbr(\bbh))\big) \right] \!\le\! 0~\forall~ s\!=\!1,\!...\!, S,~~~  \bbr(\bbh) \in \ccalR. \nonumber%
\end{alignat}
We stress in \eqref{eq_problem_resource} that the objective function $f(\bbh, \bbr(\bbh))$, the constraint functions $\{c_s\big(\bbr(\bbh), f(\bbh, \bbr(\bbh))\big)\}_{s=1}^S$ and the set $\ccalR$ are not necessarily convex depending on specific applications. In fact, in most practical scenarios, they are non-convex given the complexity of FSO systems. In general, the objective function is typically complicated and the allocated resources can be both continuous and discrete, such that solving the resource allocation problem \eqref{eq_problem_resource} can be difficult. There are mainly four challenges in our concern:
\begin{enumerate}[(i)]
\item The objective function $f(\bbh, \bbr(\bbh))$ can be extremely complicated in FSO systems, yielding non-convex optimization problems---see Sections \ref{poweradaption} and \ref{relayselection}.
\item The imposed constraints $\{c_s\big(\bbr(\bbh), f(\bbh, \bbr(\bbh))\big)\}_{s=1}^S$ are difficult to address, resulting in failures of conventional optimization algorithms---see Section \ref{powerrelay}.
\item The variable to be optimized $\bbr(\bbh)$ is a function of the channel state information $\bbh$ and consequently is infinitely dimensional.
\item FSO systems are sophisticated due to the complexity of optical equipments. Mathematical models $f(\bbh, \bbr(\bbh))$ characterizing these systems may be unknown or inaccurate such that model-based algorithms are inapplicable or suffer from inevitable degradations.
\end{enumerate}

In what follows, we first propose a model-based algorithm that solves the problem \eqref{eq_problem_resource} exactly without any approximation (Section \ref{sec_sdg}). We proceed to develop a model-free algorithm via deep learning that solves \eqref{eq_problem_resource} with only system observations, where the knowledge of system models is not required (Section \ref{sec_pddl}). Before proceeding, we illustrate in the following subsections how the general problem framework in \eqref{eq_problem_resource} represents a variety of optimal resource allocation problems in FSO communications.

\subsection{Power Adaptation} \label{poweradaption}

We consider the transmission power allocation in a Radio on Free Space Optics (RoFSO) system. As a universal platform for heterogeneous wireless services, it transmits RF signals through FSO links in optical networks \cite{zhou2015optical}. The Dense Wavelength Division Multiplexing (DWDM) RoFSO system allows simultaneous transmissions of multiple signals to increase the transmission capacity. In particular, multimedia RF signals are accessed into the RoFSO system and placed on multiple optical wavelength carriers with optoelectronic devices, and then transmitted into free space. At the receiver, optical signals are received through FSO channels and transferred back to RF signals for users---see Fig. \ref{fig_rofso}.

Based on the CSI, adaptive powers are assigned to different wavelengths to maximize the total channel capacity. Assume that there are $N$ optical wavelength carriers with non-overlapping space between each other. The CSI is represented by the vector $\bbh = [h_1,\ldots,h_N]^\top \in \mathbb{R}^N$, where $h_i$ is the CSI of $i$-th wavelength channel. The allocated power to signal transmitted on $i$-th wavelength is based upon the observed CSI $\bbh$ via a power allocation policy $p_i(\bbh)$. Given the collection of power allocations $\bbr(\bbh)=[p_1(\bbh),\ldots,p_N(\bbh)]^\top \in \mathbb{R}^N$ and the CSI $\bbh$, the channel capacity $C_i(\bbh, \bbr(\bbh))$ achieved on $i$-th wavelength is  \cite{zhou2015optical}
\begin{align} \label{eq_capacity}
&C_i(\bbh, \bbr(\bbh))=C_i(h_i, p_i(\bbh))= \log \!\Big( \!1\!+\! \frac{\frac{1}{2}(OMI\cdot m_prp_i(\bbh)h_i)^2}{RIN \cdot (rp_i(\bbh)h_i)^2+2em_p^{2+F}rp_i(\bbh)h_i\!+\!\frac{4KT}{R_f}} \Big)
\end{align}
with $OMI$ the optical modulation index, $RIN$ the relative intensity noise, $m_p$ the photodiode gain, $r$ the photodiode responsivity, $e$ the electric charge, $F$ the excess noise factor, $K$ the Boltzmann's constant, $T$ the temperature and $R_f$ the photodiode resistance. We consider the weight vector $\bbomega \!=\! [\omega_1,\ldots,\omega_N]^\top \in \mathbb{R}^N$ to represent priorities of different wireless services, and the objective function is the weighted sum of channel capacities over $N$ wavelengths   
\begin{align} \label{eq_capacity1}
\mathbb{E}_\bbh \left[f(\bbh, \bbr(\bbh))\right]= \sum_{i=1}^N \omega_i \mathbb{E}_\bbh\left[C_i(\bbh, \bbr(\bbh))\right].
\end{align}
The RoFSO system is constrained by a total power limitation $P_t$ at the base station, i.e., $\mathbb{E}_\bbh \left[ c(\bbr(\bbh)) \right] = \mathbb{E}_\bbh\Big[\sum_{i=1}^N p_i(\bbh)\Big] - P_t \le 0$, and a peak power limitation $P_s$ for each carrier to ensure eye safety, i.e., $\ccalR = [0,P_s]^N$.

\subsection{Relay Selection} \label{relayselection}

We consider the relay-assisted FSO network, in which the transmitter communicates with the receiver through intermediate hops \cite{hassan2016statistical}. In particular, assume that there are $N$ hops where each hop consists of $M$ parallel relays. The transmitter sends the optical signal to a selected relay at $1$-st hop. The latter amplifies the received signal and then transmits it to a selected relay at $2$-nd hop. The process performs successfully through $N$ hops until the receiver---See Fig. \ref{fig_relay1}. Based on the CSI, different relays are selected at different hops to maximize the channel capacity. We denote by $\bbh \in \reals^{(M \!\times\! N \!+\! 2)\!\times\! M}$ the CSI between the transmitter, relays and the receiver, and the matrix $\bbr(\bbh) \!=\! [\bbalpha_1(\bbh), \ldots, \bbalpha_N(\bbh)]^\top \in \{ 0, 1 \}^{N \times M}$ the selected relays, where each $\bbalpha_i(\bbh) = [\alpha_{i1}(\bbh),\ldots,\alpha_{iM}(\bbh)]^\top \in \{ 0, 1 \}^{M}$ is a $M$-dimensional vector with $\alpha_{ij}(\bbh)=1$ if $j$-th relay is selected at $i$-th hop and $\alpha_{ij}(\bbh)=0$ otherwise. The relay-assisted channel capacity is \cite{hassan2017statistical}
\begin{align} \label{eq_capacity35}
&C_{j_1 \ldots j_N} \!(\bbh)\!=\!\!\frac{T_f B}{\epsilon}\! \log \!\Big( 1\!+\! \Big( \prod_{i=0}^N\! \Big( 1\!+\!\frac{1}{P h_{j_{i}j_{i+1}} \frac{R}{e \Delta f}} \Big)\!-\!1 \Big)^{-\!1} \Big)
\end{align}
which assumes that $j_i$-th relay is selected at $i$-th hop and $h_{j_{i}j_{i+1}}$ is the CSI between $j_{i}$-th relay at $i$-th hop and $j_{i+1}$-th relay at $(i+1)$-th hop, where $j_0 \!=\! j_{N\!+\!1} \!=\!1$ represent the transmitter and the receiver. Here, $T_f$ is the frame duration, $B$ the bandwidth, $P$ the transmission power, $R$ the photodetector sensitivity, $e$ the electric charge, $\Delta f$ the noise equivalent bandwidth, $\epsilon=1$ for the full-duplex relay and $\epsilon=2$ for the half-duplex relay. The objective function is then given by
\begin{align} \label{eq_capacity3}
&\mathbb{E}_\bbh \!\left[f(\bbh, \bbr(\bbh))\right]\!=\!\mathbb{E}_\bbh\Big[ \sum_{j_N = 1}^M \!\cdots\! \sum_{j_1 = 1}^M\! \Big(\prod_{i=\!1}^N\! \alpha_{i j_i}(\bbh)\! \Big) C_{j_1 \ldots j_N} (\bbh)\Big].
\end{align}
There are $N$ constraints on the selected relays $\bbr(\bbh)$. That is only one relay can be selected at each hop, i.e., $\ccalR \!=\! \Big\{\! \{0,1\}^{N\times M} | \sum_{j=1}^M \alpha_{ij}(\bbh) \leq 1,\forall~i\!=\!1,...,N \!\Big\}$. There is no additional stochastic constraint in this example.

\subsection{Joint Power and Relay Allocation} \label{powerrelay}

The resource allocation problem becomes more complicated when considering the joint power and relay allocation, as seen in the FSO fronthaul network \cite{hassan2017delay, hassan2017statistical}. As one of cloud radio access network (C-RAN) architectures, it provides high rates, low latency and flexible constructions for 5G wireless networks. In particular, the system consists of remote radio heads (RRHs), aggregation nodes (ANs) and the baseband unit (BBU). The RRHs transmit optical signals with orthogonal optical carriers through free space to the selected ANs. The latter collect received signals and then forward the aggregated signal to the BBU through high speed optical fiber---See Fig. \ref{fronthaul}. Based on the CSI, different ANs are selected at different RRHs and adaptive powers are assigned to different optical carriers at each RRH. Assume there are $L$ optical carriers, $N$ RRHs, $M$ ANs and one BBU. The CSI is represented by $\bbh = \{\bbh_{ij}\}_{ij}$ for all $i\!=\!1,\!...\!,N$ and $j\!=\!1,\!...\!,M$, where each vector $\bbh_{ij}\!=\![h^1_{ij},...,h^L_{ij}]^\top \!\in\! \mathbb{R}^{L}$ is the CSI of $L$ optical carriers between $i$-th RRHs and $j$-th AN. The allocated resources $\bbr(\bbh)\!=\!\{ \bbp_{ij}(\bbh), \alpha_{ij}(\bbh)\}_{ij}$ contain assigned powers and selected ANs, where $\bbp_{ij}(\bbh) = [p_{ij}^1(\bbh),...,p_{ij}^L(\bbh)]^\top \!\in\! \mathbb{R}^{L}$ are powers assigned to $L$ optical carriers in the link between $i$-th RRH and $j$-th AN, and $\alpha_{ij}(\bbh) \in \{ 0,1 \}$ is the indicator being one if $j$-th AN is selected at $i$-th RRH and zero otherwise. The channel capacity between $i$-th RRH and $j$-th AN is \cite{hassan2017delay}
\begin{align} \label{eq_capacity4}
&C_{ij}(\bbh, \bbr(\bbh)) =\sum_{\ell=1}^L \omega_{\ell} \frac{T_f \!B}{\epsilon} \log \Big( 1+ p_{ij}^\ell(\bbh) h_{ij}^\ell \frac{R}{e \Delta f} \Big)
\end{align}
with $\bbomega = [\omega_1,\ldots,\omega_L]^\top \in \mathbb{R}^L$ the priorities of optical carriers, $T_f$ the frame duration, $B$ the bandwidth, $R$ the photodetector sensitivity, $e$ the electric charge and $\Delta f$ the noise equivalent bandwidth. The objective function is the sum-capacity over $N$ RRHs
\begin{align} \label{eq_fronthaulFSO_capacity}
&\mathbb{E}_\bbh \left[f(\bbh, \bbr(\bbh))\right]=\mathbb{E}_\bbh\!\Big[\sum_{i=1}^N\!  \sum_{j=1}^M \alpha_{ij}(\bbh)C_{ij}(\bbh, \bbr(\bbh))\Big].
\end{align}
There are $2N$ constraints for the allocated powers, $N$ constraints for the selected ANs and additional $M$ constraints for the aggregated data at ANs: (i) the total power limitation $P_t$ and the peak power limitation $P_s$ at each RRH as in Section \ref{poweradaption}; (ii) only one AN can be selected at each RRH as in Section \ref{relayselection}; (iii) the aggregated data traffic shall not exceed the maximal capacity $C_t$ of optical fiber at each AN to avoid data congestion. Therefore, we have
\begin{subequations}\label{eq_fronthaulFSO_constraints}
\begin{align}
&\mathbb{E}_\bbh\!\Big[\sum_{\ell=1}^L\! p_{ij}^\ell(\bbh)\Big] - P_t \le\! 0,~\forall~i =1,...,N, j=1,...,M,\\
&\mathbb{E}_\bbh\!\Big[\sum_{i=1}^N C_{ij}(\bbh, \bbr(\bbh))\Big] \!- C_t \le\! 0, ~\forall~ j\!=\!1,\ldots,M,\\
\label{eq_fronthaulFSO_constraintsc}
&\!\ccalR\! =\! \Big\{ [0,P_s]^{N\times M\times L} \!\times\! \{ 0,\!1 \}^{N\times M} | \sum_{j=1}^M\! \alpha_{ij}(\bbh) \leq 1,\forall~i =1,...,N \Big\} .
\end{align}
\end{subequations}

%%%%%%%%%%%%%%%%%%%%%%%%%%%%%%%%%%%%%%%%%%%%%%%%%%%%%%%%%%%%%%%%%%%%%%%%%%%%%%%%%%%%%%%%%%%%%%%%%%%%%%%%%%%%%%%%%%%%%%%%%%%%%%%%%%%%%%%%%%%%%%%%%%%%%%%%%%%%%%%%%%%%%%%%%%A  S  S  U  M  P  T  I  O  N %%%%%%%%%%%%%%%%%%%%%%%%%%%%%%%%%%%%%%%%%%%%%%%%%%%%%%%%%%%%%%%%%%%%%%%%%%%%%%%%%%%%%%%%%%%%%%%%%%%%%%%%%%%%%%%%%%%%%%%%%%%%%%%%%%%%%%%%%%%%%%%%%%%%%%%%%%%%%%%%%%%%%%%%%%%%%%%%%%%%%%%%%%%%%%%%

%!TEX root = mainOp.tex
%%%%%%%%%%%%%%%%%%%%%%%%%%%%%
%%% SECTION : Stability   %%%
%%%%%%%%%%%%%%%%%%%%%%%%%%%%%

\section{Stochastic Dual Gradient Algorithm}\label{sec_sdg}

In this section, we first address three primary challenges (i)-(iii) outlined in Section \ref{sec:problem} by working in the dual domain. In particular, by establishing a null duality gap for \eqref{eq_problem_resource}, we present the Stochastic Dual Gradient (SDG) algorithm that finds exact solutions despite the non-convexity, constraints, and infinite dimensionality. For the purposes of developing the SDG algorithm, we initially ignore challenge (iv) and assume mathematical models established for FSO systems are given and accurate. For instance, in the RoFSO system we assume the channel capacity function $C_i(\bbh, \bbr(\bbh))$ in \eqref{eq_capacity} characterizes the RoFSO system accurately.

With a set of convex or non-convex constraints, it is natural to consider working in the dual domain. By introducing the dual variables $\bblambda = [\lambda_1, \ldots, \lambda_S]^\top \in \mathbb{R}_+^{S}$ that correspond to $S$ constraints, the Lagrangian of problem \eqref{eq_problem_resource} is given by
\begin{align} \label{eq_lagran}
&\mathcal{L}(\bbr(\bbh),\bblambda) = \mathbb{E}_\bbh\! \left[f(\bbh, \bbr(\bbh))\right] - \sum_{s=1}^S \lambda_s \mathbb{E}_\bbh \! \left[ c_s\left(\bbr(\bbh), f(\bbh, \bbr(\bbh))\right) \right].
\end{align}
Each constraint in \eqref{eq_problem_resource} shows as a penalty in \eqref{eq_lagran}, where the violation is penalized (weighted by a dual variable). We define the dual function as the maximum of Lagrangian
\begin{equation} \label{dualfunc}
\begin{split}
\mathcal{D}(\bblambda) & = \max_{\bbr(\bbh)\in \mathcal{R}} \mathcal{L}(\bbr(\bbh),\bblambda).
\end{split}
\end{equation}
The problem \eqref{dualfunc} is unconstrained such that conventional optimization algorithms can be used. With dual variables involved, it has been proved that $\ccalD(\bblambda) \ge \mathbb{P}$ holds for any $\bblambda$. This result motivates the development of dual problem, that is to find $\bblambda^*$ minimizing the dual function as
\begin{equation} \label{eq_dualprob}
\begin{split}
\mathbb{D} := \min_{\bblambda \ge 0} \mathcal{D}(\bblambda)=\min_{\bblambda \ge 0} \max_{\bbr(\bbh)\in \mathcal{R}} \mathcal{L}(\bbr(\bbh),\bblambda).
\end{split}
\end{equation}
The optimal solution $\mathbb{D}$ for \eqref{eq_dualprob} can be viewed as the best approximation of $\mathbb{P}$ when handling constraints as penalties. However, it is still unclear how much the difference between $\mathbb{D}$ and $\mathbb{P}$ is and further how to develop an algorithm to solve the alternative min-max problem \eqref{eq_dualprob}. We consider these issues in following subsections.

\subsection{Null Duality Gap}

For the ideal scenario, one would expect the difference $\mathbb{D}-\mathbb{P}$, referred to as the duality gap, to be zero. As such, we can solve the general resource allocation problem \eqref{eq_problem_resource} by solving its associated dual problem \eqref{eq_dualprob} without loss of optimality. It is well-known that the null duality gap holds for convex optimization problems, which the problem \eqref{eq_problem_resource} rarely leads to due to complicated objective and constraint functions. Despite its possible non-convexity, we show that the problem does have the null duality gap in the following theorem.

\begin{theorem}\label{theorem1}
Consider the stochastic optimization problem \eqref{eq_problem_resource} and its associated dual problem \eqref{eq_dualprob}. Let $\mathbb{P}$ be the optimal solution of \eqref{eq_problem_resource} and $\mathbb{D}$ be the optimal solution of \eqref{eq_dualprob}. Assume that there exists a feasible point $\bbr_0$ satisfying all constraints with strict inequality, then the strong duality holds that $\mathbb{P}=\mathbb{D}$.
\end{theorem}
\begin{proof}
The objective and the constraints in \eqref{eq_problem_resource} are in expectation with respect to the CSI $\bbh$, which is instantiated from the probability distribution $m(\bbh)$. We can consider \eqref{eq_problem_resource} as a particular realization of the sparse functional program \cite{chamon2020functional}. Since the distribution $m(\bbh)$ that characterizes the FSO channel is continuous, $\bbh$ takes values in a dense set of the domain. Considering this observation together with the assumption that there exists a solution satisfying all constraints with strict inequality, the results of Theorem 1 in \cite{chamon2020functional} claim the strong duality gap $\mathbb{P} = \mathbb{D}$.
\end{proof}

Theorem \ref{theorem1} states that the optimization problem \eqref{eq_problem_resource} has the null duality gap $\mathbb{D}-\mathbb{P}=0$ even if it is non-convex, where the strict feasibility assumption is mild in practice. We can then solve \eqref{eq_problem_resource} by solving the unconstrained dual problem \eqref{eq_dualprob} alternatively without loss of optimality.

\subsection{Primal-Dual Update}

We propose the SDG algorithm based on the above analysis, which iteratively searches for the optimal dual variables $\bblambda^*$ starting from an initial iterate $\bblambda^0$, to derive the corresponding optimal resource allocation policy $\bbr^*(\bbh)$. To be more precise, the SDG consists of two steps over an iteration index $k$. The primal step updates the primal variables $\bbr(\bbh)$ given the current dual variables $\bblambda^k$, while the dual step updates the dual variables $\bblambda$ given the updated $\bbr^{k+1}(\bbh)$. Details are formally introduced below.

(1) \emph{Primal step.} At $k$-th iteration given the dual variables $\bblambda^{k}$ and the CSI $\bbh$, we update the primal variables by maximizing the Lagrangian as
 \begin{align} \label{eq_priup}
\bbr^{k+1}(\bbh)&= \argmax_{\bbr(\bbh) \in \mathcal{R}} \mathcal{L}\big(\bbr(\bbh),\bblambda^{k}\big)= \argmax_{\bbr(\bbh) \in \mathcal{R}}\mathbb{E}_\bbh \big[f(\bbh, \bbr(\bbh))\big] \!\!- \!\!\sum_{s=1}^S\! \lambda_s^k \mathbb{E}_\bbh \! \big[ c_s\big(\bbr(\bbh), f(\bbh, \bbr(\bbh))\big) \!\big]\nonumber\\
&= \argmax_{\bbr(\bbh) \in \mathcal{R}}f(\bbh, \bbr(\bbh)) \!- \!\sum_{s=1}^S\! \lambda_s^k c_s\left(\bbr(\bbh), f(\bbh, \bbr(\bbh))\right)
\end{align}
where the last equality is because the expectation is automatically maximized if it is maximized at each sample $\bbh$. In practice, \eqref{eq_priup} can usually be simplified based on specific system models. For example, in the RoFSO system, both the objective and the constraints separate the use of components $p_1(\bbh), \ldots, p_N(\bbh)$ in $\bbr(\bbh)$ and $h_1,\ldots,h_N$ in $\bbh$ with no coupling between them. In this context, solving \eqref{eq_priup} is equivalent to solving $N$ scalar sub-problems that update each component $p_i(\bbh)$ separately as $p_{i}^{k+1}(\bbh)= \argmax_{p_i(\bbh) \in [0,P_S]} \omega_i C_i(h_i, p_i(\bbh))-p_i(\bbh)$ for all $i=1,\ldots,N$.

%%%%%%%%%%%%%%%%%%%%%%%%%%%%%%%%%%%%%%%%%%%%%%%%%%%%%%%%%%%%%%%%
%%%%   A   L   G   O   R   I   T   H   M   %%%%%%%%%%%%%%%%%%%%%
%%%%%%%%%%%%%%%%%%%%%%%%%%%%%%%%%%%%%%%%%%%%%%%%%%%%%%%%%%%%%%%%
{\linespread{1.3}
\begin{algorithm}[t] \begin{algorithmic}[1]
\STATE \textbf{Input:} The objective function $f(\bbh, \bbr(\bbh))$, the constraints $\{ c_s\big(\bbr(\bbh), f(\bbh, \bbr(\bbh))\big)\}_{s=1}^S$, the CSI $\bbh$ and the initial dual variables $\bblambda^0$
\FOR [main loop]{$k = 0,1,2,\hdots$}
      \STATE Update the primal variables $\bbr^{k+1}(\bbh)$ by \eqref{eq_priup}
      \STATE $\bbr^{k+1}(\bbh)=\argmax_{\bbr(\bbh) \in \mathcal{R}}f(\bbh, \bbr(\bbh)) \!- \!\sum_{s=1}^S\! \lambda_s^k c_s\left(\bbr(\bbh), f(\bbh, \bbr(\bbh))\right)$
	\STATE Update the dual variables $\bblambda^{k+1}$ by \eqref{eq_dualup} \\
	\FOR [main loop]{$s = 1,\ldots,S$}
    	\STATE $\lambda_s^{k+1} =\left[ \lambda_s^{k}-\eta^k c_s\!\left(\bbr^{k+1}(\bbh), f(\bbh, \bbr^{k+1}(\bbh))\right) \right]_+ $
    	\ENDFOR
\ENDFOR

\end{algorithmic}
\caption{Stochastic Dual Gradient Algorithm}\label{alg:learning1} \end{algorithm}}

(2) \emph{Dual step.} Given the updated $\bbr^{k+1}(\bbh)$ from the primal step (1), we perform the dual gradient descent to update $\bblambda^{k}$ as
\begin{equation} \label{eq_dualup}
\begin{split}
\lambda^{k+1}_s &= \left[ \lambda_s^{k} - \eta^k \nabla_{\lambda_s} \ccalL(\bbr^{k+1}(\bbh),\bblambda^k) \right]_+ =\left[ \lambda_s^{k}-\eta^k c_s\!\left(\bbr^{k+1}(\bbh), f(\bbh, \bbr^{k+1}(\bbh))\right) \right]_+
\end{split}
\end{equation}
for all $s=1,\ldots,S$, where $\eta^k$ is the dual step-size at iteration $k$ and $[\cdot]_+ = \max(\cdot, 0)$ is due to the non-negativity of the dual variables $\bblambda$.

By repeating these two steps recursively, $\bblambda^k$ converges to the optimal values $\bblambda^*$ as $k$ increases \cite{bottou2012stochastic}. Due to the null duality gap, the optimal solution $\bbr^*(\bbh)$ can be obtained from $\bblambda^*$ as
 \begin{align} \label{eq:SDGoptimalCompute}
\bbr^{*}(\bbh)\!=\! \argmax_{\bbr(\bbh) \in \mathcal{R}}\!f(\bbh, \bbr(\bbh)) \!- \!\sum_{s=1}^S\! \lambda_s^* c_s \big(\bbr(\bbh), f(\bbh, \bbr(\bbh))\big).
\end{align}
Algorithm \ref{alg:learning1} summarizes the SDG algorithm.

With accurate system models, the SDG algorithm solves the problem \eqref{eq_problem_resource} perfectly in theory with no relaxation or approximation. However, there exist several problems with respect to its practical implementation. For one thing, in the primal step of the SDG, there is no closed-form solution of \eqref{eq_priup} to compute optimal $\bbr^{k+1}(\bbh)$. Similarly, even after the algorithm converges, real time execution of $\bbr^*(\bbh)$ needs to numerically solve \eqref{eq:SDGoptimalCompute}. Therefore, it may require certain computational complexity. For another, we recall the difficulty of obtaining accurate system models in FSO systems as stated in challenge (iv) of Section \ref{sec:problem}. The SDG algorithm heavily relies on system models, i.e., we need accurate knowledge of the objective functions $f(\bbh, \bbr(\bbh))$ and constraint functions $c_s \big(\bbr(\bbh),\! f(\bbh, \bbr(\bbh))\big)$ to perform algorithm. This may not be available given a FSO system that is new, unfamiliar or complex. Furthermore, existing models do not always capture the true physical performance in practice leading to inevitable model errors. These motivate the development of a low-complexity and model-free learning-based algorithm to solve resource allocation problems, as we introduce in the following section.

%!TEX root = mainOp.tex
%%%%%%%%%%%%%%%%%%%%%%%%%%%%%
%%% SECTION : Stability   %%%
%%%%%%%%%%%%%%%%%%%%%%%%%%%%%

\section{Primal-Dual Deep Learning Algorithm}\label{sec_pddl}

To handle above limitations, we develop the model-free Primal-Dual Deep Learning (PDDL) algorithm  based on the SDG algorithm. The implementation of the PDDL requires only observed values of the FSO system (e.g., the observed channel capacity and CSI) instead of mathematical system models. We begin by noticing the problem \eqref{eq_problem_resource} shares the same structure as the statistical learning problem. The latter inspires us to introduce a parametrization $\bbtheta\in \mathbb{R}^q$ to represent the resource allocation policy as $\bbr(\bbh) = \bbPhi(\bbh, \bbtheta)$. Substituting this representation into \eqref{eq_problem_resource} yields
\begin{alignat}{3} \label{learning_prob}
 \mathbb{P}_{\bbtheta}:= &  \max_{\bbtheta} \ && \mathbb{E}_\bbh \left[f(\bbh, \bbPhi(\bbh, \bbtheta))\right] ,             \\
        &  \st                    \ && \!\mathbb{E}_\bbh \! \left[ c_s\big(\bbPhi(\bbh, \bbtheta), f(\bbh, \bbPhi(\bbh, \bbtheta))\big) \right] \!\le\! 0~\forall~ s=1,..., S, ~~ \bbtheta \in \Theta    \nonumber %
\end{alignat}
where $\Theta$ is the parametrization set satisfying $\bbPhi(\bbh, \bbtheta) \in \ccalR$. Then, the goal becomes to learn the optimal function $\bbPhi^*(\bbh, \bbtheta^*)$ by finding the optimal parametrization $\bbtheta^*$  that maximizes the objective while satisfying prescribed constraints.

\subsection{Near-Universal Parametrization}

The parametrization in \eqref{learning_prob} inevitably introduces a loss of optimality since resource allocation functions are restricted to those adhered to the form of $\bbr(\bbh) = \bbPhi(\bbh, \bbtheta)$. For example, a linear parametrization $\bbPhi(\bbh, \bbtheta)= \bbtheta^\top \bbh$ can never represent any nonlinear resource allocation policy. A good choice of $\bbPhi(\bbh, \bbtheta)$ should provide an accurate approximation for almost all functions in $\ccalR$ by changing parameters $\bbtheta$, and thus can model the space of allowable resource allocation policies to guarantee the learning performance. To quantify such function representation ability, we define the near-universal parametrization as follows.
\begin{definition}[Near-universal parametrization]
For any $\epsilon \ge 0$, the parametrization $\bbPhi(\bbh, \bbtheta)$ is $\epsilon$-universal if for any $\bbr(\bbh) \in \ccalR$, there exists a set of parameters $\bbtheta \in \Theta$ such that
\begin{equation}
\mathbb{E}_\bbh \left[ \|\bbr(\bbh) - \bbPhi(\bbh, \bbtheta)\|_\infty \right] \le \epsilon.
\end{equation}
\end{definition} 
\noindent The universal property has been found in a number of learning architectures, e.g., radial basis function networks \cite{park1991universal}, reproducing kernel Hilbert spaces \cite{sriperumbudur2010relation} and deep neural networks \cite{hornik1991approximation}. 

Deep Neural Networks (DNNs) in particular are well-suited candidates that exhibit the universal function approximation ability and achieve successes in various practical problems. DNNs are information processing architectures consisting of multiple layers, each of which comprises linear operations and pointwise nonlinearities. Specifically, consider a DNN with $L$ layers. At layer $\ell$, we have the input feature $\bbx_{\ell-1} \in \mathbb{R}^{n_{\ell-1}}$ with $n_{\ell-1}$ the number of hidden units at layer $(\ell\!-\!1)$. This feature is processed by the linear operation $\bbPi_\ell \in \mathbb{R}^{n_{\ell}\times n_{\ell-1}}$ to obtain the higher-level feature $\bbu_{\ell} \in \mathbb{R}^{n_{\ell}}$. The latter is passed through a pointwise nonlinearity $\sigma(\cdot): \mathbb{R} \to \mathbb{R}$ to generate the output feature $\bbx_{\ell} = \sigma\big( \bbPi_\ell \bbx_{\ell-1} \big)$. The output at layer $\ell$ is again taken as the input at layer $(\ell+1)$, and the process repeats recursively until the final layer $L$---see Fig. \ref{fig.evRecMain}. In our case, the input of the DNN is the instantaneous CSI $\bbx_0=\bbh$ and the output is $\bbx_L = \bbPhi(\bbh, \bbtheta)$. The parametrization $\bbtheta \in \mathbb{R}^q$ are the weights of linear operations $\bbPi_1, \ldots, \bbPi_L$, where $q = \sum_{\ell=0}^{L-1} n_\ell n_{\ell+1}$ is determined by feature dimensions $n_0, \ldots, n_L$. Common examples for the nonlinearity are the absolute value, the ReLU, the sigmoid function, etc. We then verify its near-universal property as follows.

\begin{theorem}\cite[Theorem 2.2]{hornik1991approximation}\label{theorem2}
Let $m(\bbh)$ be the distribution of the channel state information $\bbh$ and $\ccalR$ be the considered set of measurable functions. For a DNN with arbitrarily large number of layers and arbitrarily large layer sizes, it is dense in probability in $\ccalR$, i.e., for any function $\bbr(\bbh) \in \ccalR$ and $\epsilon > 0$, there exists $L$, $\{ n_1,\ldots,n_L \}$ and $\bbtheta \in \mathbb{R}^q$ such that 
\begin{equation} \label{mainthm2}
\begin{split}
m\big( \{ \bbh: \| \bbPhi(\bbh, \bbtheta) - \bbr(\bbh) \|_\infty > \epsilon \} \big) < \epsilon.
\end{split}
\end{equation}
\end{theorem}
Theorem \ref{theorem2} states that DNNs can approximate functions in the considered set with arbitrarily small error $\epsilon$ by increasing the number of layers $L$ and layer sizes $\{n_\ell\}_{\ell=1}^L$. Therefore, the parametrization loss $\mathbb{P} - \mathbb{P}_{\bbtheta}$ can be sufficiently small by learning with the DNN parametrization.

%%%%%%%%%%%%%%%%%%%%%%%%%%%%%%%%%%%%%%%%%%%%%%%%%%%%%%%%%%%%%%%%%%%%%%%
%%%   F   I   G   U   R   E   %%%%%%%%%%%%%%%%%%%%%%%%%%%%%%%%%%%%%%%%%%%%%%%%%%%%%%%%%
%%%%%%%%%%%%%%%%%%%%%%%%%%%%%%%%%%%%%%%%%%%%%%%%%%%%%%%%%%%%%%%%%%%%%%%
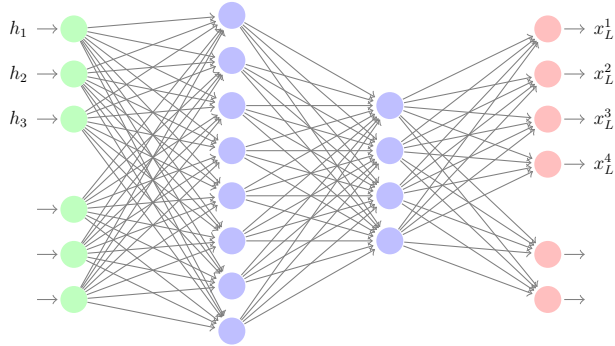
\begin{figure*}
\centering
\def\layersep{3.5cm}

\begin{tikzpicture}[shorten >=1pt,->,draw=black!50, node distance=.8*\layersep, scale=0.6, transform shape]
    \tikzstyle{every pin edge}=[<-,shorten <=1pt]
    \tikzstyle{neuron}=[circle,fill=black!25,minimum size=17pt,inner sep=0pt]
    \tikzstyle{tneuron}=[circle,line width=1.2pt,minimum size=17pt,inner sep=0pt]
    \tikzstyle{input neuron}=[neuron, fill=green!25];
    \tikzstyle{output neuron}=[neuron, fill=red!25];
    \tikzstyle{hidden neuron}=[tneuron, fill=blue!25];
    \tikzstyle{annot} = [text width=4em, text centered]

    % Draw the input layer nodes
  %  \foreach \name / \y in {1,...,6}
    % This is the same as writing \foreach \name / \y in {1/1,2/2,3/3,4/4}
        \node[input neuron, pin=left: $h_{1}$] (I-1) at (0,-1) {};
        \node[input neuron, pin=left: $h_{2}$] (I-2) at (0,-2) {};
        \node[input neuron, pin=left: $h_{3}$] (I-3) at (0,-3) {};
        
        \node[input neuron, pin=left:] (I-4) at (0,-5) {};
        \node[input neuron, pin=left:] (I-5) at (0,-6) {};
        \node[input neuron, pin=left:] (I-6) at (0,-7) {};
        
% Draw the hidden layer nodes
    \foreach \name / \y in {1,2,3,4,5,6,7,8}
        \path[yshift=0.3cm]
            node[hidden neuron] (H1-\name) at (\layersep,-\y cm) {};

    % Draw the hidden layer nodes
    \foreach \name / \y in {1,...,4}
        \path[yshift=0.3cm]
            node[hidden neuron] (H-\name) at (2*\layersep,-2 cm-\y cm) {};

    % Draw the output layer node
    	\node[output neuron,pin={[pin edge={->}]right:$x_L^1$}] (O1) at (3*\layersep,-1 cm) {};
	\node[output neuron,pin={[pin edge={->}]right:$x_L^2$}] (O2) at (3*\layersep,-2 cm) {};
	\node[output neuron,pin={[pin edge={->}]right:$x_L^3$}] (O3) at (3*\layersep,-3 cm) {};
	\node[output neuron,pin={[pin edge={->}]right:$x_L^4$}] (O4) at (3*\layersep,-4 cm) {};
	\node[output neuron,pin={[pin edge={->}]right:}] (O5) at (3*\layersep,-6 cm) {};
	\node[output neuron,pin={[pin edge={->}]right:}] (O6) at (3*\layersep,-7 cm) {};
	%\node[output neuron,pin={[pin edge={->}]right:$\sigma^i$}, right of=H-3] (O2) {};

    % Connect every node in the input layer with every node in the
    % hidden layer.
    \foreach \source in {1,...,6}
        \foreach \dest in {1,...,8}
            \path (I-\source) edge (H1-\dest);
            
     \foreach \source in {1,...,8}
        \foreach \dest in {1,...,4}
            \path (H1-\source) edge (H-\dest);

    % Connect every node in the hidden layer with the output layer
    \foreach \source in {1,...,4}
        \path (H-\source) edge (O1);
        
     \foreach \source in {1,...,4}
        \path (H-\source) edge (O2);
        
            \foreach \source in {1,...,4}
        \path (H-\source) edge (O3);
        
     \foreach \source in {1,...,4}
        \path (H-\source) edge (O4);
        
            \foreach \source in {1,...,4}
        \path (H-\source) edge (O5);
        
     \foreach \source in {1,...,4}
        \path (H-\source) edge (O6);

    % Annotate the layers
  %  \node[annot,above of=H-1, node distance=1cm] (hl) {Hidden layer};
  %  \node[annot,left of=hl] {Input layer};
   % \node[annot,right of=hl] {Output layer};
\end{tikzpicture}

% End of code
\caption{Deep neural network with $2$ hidden layers. The channel state information $\bbh$ is fed into the input units (green nodes), processed by the hidden units (blue nodes) and output in the last-layer units (red nodes).}
\label{fig.evRecMain}
\end{figure*}
%%%%%%%%%%%%%%%%%%%%%%%%%%%%%%%%%%%%%%%%%%%%%%%%%%%%%%%%%%%%%%%%%%%%%%%
%%%%%%%%%%%%%%%%%%%%%%%%%%%%%%%%%%%%%%%%%%%%%%%%%%%%%%%%%%%%%%%%%%%%%%%
%%%%%%%%%%%%%%%%%%%%%%%%%%%%%%%%%%%%%%%%%%%%%%%%%%%%%%%%%%%%%%%%%%%%%%%

\subsection{Primal-Dual Learning}

We now develop an analogous dual-domain learning method to find the optimal parametrization $\bbtheta^*$. Similar as the unparameterized problem, we begin by formulating the Lagrangian of \eqref{learning_prob} as
\begin{align} \label{eq_lagran1}
&\mathcal{L}(\bbtheta,\bblambda)= \mathbb{E}_\bbh\! \left[f(\bbh, \bbPhi(\bbh, \bbtheta))\right] \!-\!\! \sum_{s=1}^S\! \lambda_s \mathbb{E}_\bbh \big[ c_s\big(\bbPhi(\bbh, \!\bbtheta), f(\bbh,\bbPhi(\bbh,\! \bbtheta))\big) \big].
\end{align}
The corresponding dual problem is subsequently defined as
\begin{equation} \label{eq_dual1}
\begin{split}
\mathbb{D}_{\bbtheta} := \min_{\bblambda \ge 0} \mathcal{D_{\bbtheta}}(\bblambda)=\min_{\bblambda \ge \bb0} \max_{\bbtheta\in \Theta} \mathcal{L}(\bbtheta,\bblambda).
\end{split}
\end{equation}
For the above min-max problem with the parametrization $\bbtheta$, the duality gap $\mathbb{P}_{\bbtheta}-\mathbb{D}_{\bbtheta}$ can be assumed sufficiently small due to the strong duality in Theorem \ref{theorem1} and the near-universal property of the DNN in Theorem \ref{theorem2}. Thus, we can solve \eqref{learning_prob} by solving \eqref{eq_dual1} with little loss of optimality.

We similarly develop the PDDL algorithm for solving \eqref{eq_dual1}, which updates the primal variables $\bbtheta$ with gradient ascent and the dual variables $\bblambda$ with gradient descent at each iteration $k$:

(1) \emph{Primal step.} Given the dual variables $\bblambda^{k}$, we update the primal variables $\bbtheta$ as
\begin{align} \label{eq_priup1}
&\bbtheta^{k+1} \!=\! \bbtheta^{k} \!+\! \delta^k \nabla_{\bbtheta} \mathcal{L}(\bbtheta^{k},\bblambda^{k}) \!=\! \bbtheta^{k} \!+\! \delta^k \nabla_{\bbtheta} \mathbb{E}_\bbh\! \big[ f(\bbh,\! \bbPhi(\bbh, \!\bbtheta))\! -\!\! \sum_{s=1}^S\! \lambda_s c_s\big(\bbPhi(\bbh,\! \bbtheta),\! f(\bbh, \!\bbPhi(\bbh,\! \bbtheta))\!\big)\! \big]
\end{align}
where $\delta^k$ is the primal step-size, and the last equation is due to the linearity of the expectation.

(2) \emph{Dual step.} Given the updated $\bbtheta^{k+1}$ from step (1), the dual variables $\bblambda$ is updated as
\begin{equation} \label{eq_dualup1}
\begin{split}
\lambda^{k\!+\!1}_s \!=\!\Big[ \lambda_s^{k}\!-\!\eta^k \mathbb{E}_\bbh \big[ c_s\big(\bbPhi(\bbh,\! \bbtheta^{k+1}), f(\bbh, \bbPhi(\bbh, \bbtheta^{k+1}))\big)\big] \Big]_+
\end{split}
\end{equation}
for all $s=1,\!...\!,S$, where $\eta^k$ is the dual step-size.

The PDDL algorithm learns the optimal primal and dual variables $\bbtheta^*$ and $\bblambda^*$ by recursively repeating primal and dual steps. The primal-dual method used in the parameterized problem features a closed form update in \eqref{eq_priup1}, in contrast to the computationally expensive inner maximization required in \eqref{eq_priup} of the SDG algorithm used in the unparameterized problem. Even still, direct evaluation of the primal update in \eqref{eq_priup1} requires the knowledge of system models to compute the expected gradients, generally not available in practice. However, unlike the SDG algorithm, the PPDL algorithm is capable of leveraging the so-called policy gradient method to develop a completely model-free implementation.

%%%%%%%%%%%%%%%%%%%%%%%%%%%%%%%%%%%%%%%%%%%%%%%%%%%%%%%%%%%%%%%%
%%%%   A   L   G   O   R   I   T   H   M   %%%%%%%%%%%%%%%%%%%%%
%%%%%%%%%%%%%%%%%%%%%%%%%%%%%%%%%%%%%%%%%%%%%%%%%%%%%%%%%%%%%%%%
{\linespread{1.3}
\begin{algorithm}[t] \begin{algorithmic}[1]
\STATE \textbf{Input:} Initial primal and dual variables $\bbtheta^0, \bblambda^0$
\FOR [main loop]{$k = 0,1,2,\hdots$}
      \STATE Draw CSI samples $\{ \bbh_\tau \}_{\tau=1}^\ccalT$, and get corresponding allocated resources $\{ \bbr_\tau \}_{\tau=1}^\ccalT$ according to DNN outputs $\{\bbPhi(\bbh_\tau, \bbtheta^k)\}_{\tau=1}^\ccalT$ and policy distributions $\{\pi_{\bbh_\tau,\bbtheta^k}(\bbr)\}_{\tau=1}^\ccalT$
      \STATE Obtain observations of the objective function $\{ f( \bbh_\tau, \bbr_\tau)\}_{\tau=1}^\ccalT$ at current samples $\{ \bbh_\tau \}_{\tau=1}^\ccalT$
      \STATE Compute the policy gradient $\widetilde{ \nabla_{\bbtheta}} \mathcal{L}(\bbtheta^k, \bblambda^k)$ by \eqref{eq_polgrad}
      \STATE Update the primal variables $\bbtheta^{k+1} = \bbtheta^{k} + \delta^k \widetilde{ \nabla_{\bbtheta}} \mathcal{L}(\bbtheta^k,\bblambda^k) \nonumber $ [cf. \eqref{eq_priup1}]\\
	\STATE Update the dual variables $\bblambda^{k+1} \!\!=\!\! \Big[\! \bblambda^{k} \!-\! \frac{\eta^k }{\ccalT}\!\sum_{\tau \!=\! 1}^\ccalT\! c_s\!\big(\!\bbPhi(\bbh_\tau,\! \bbtheta^{k+1}\!), f(\bbh_\tau,\! \bbPhi(\bbh_\tau,\! \bbtheta^{k\!+\!1}\!)\!)\!\big)\!\big] \!\Big]_+ \nonumber $ [cf. \eqref{eq_dualup1}]
\ENDFOR

\end{algorithmic}
\caption{Primal-Dual Deep Learning Algorithm}\label{alg:learning} \end{algorithm}}

\subsection{Model-Free Policy Gradient}

Policy gradient has been developed as a practical gradient estimation method in reinforcement learning because it avoids explicit modeling of the objective function $f(\cdot)$ and the constraint functions $c_s(\cdot)$. It exploits a likelihood ratio property to compute the gradient for policy functions taking the form of $\mathbb{E}_\bbh [f(\bbh, \bbPhi(\bbh,\bbtheta))]$, where $f(\cdot)$ is unknown. Put simply, it provides a stochastic and model-free approximation for $\nabla_{\bbtheta} \mathbb{E}_\bbh [f(\bbh, \bbPhi(\bbh,\bbtheta))]$ \cite{sutton2000policy}.

In particular, we consider the policy parametrization $\bbPhi(\bbh,\bbtheta)$ as stochastic realizations drawn from a distribution with the delta density function $\pi_{\bbh,\bbtheta}(\bbr)=\delta(\bbr - \bbPhi(\bbh, \bbtheta))$. We can then rewrite the Jacobian of policy function as 
\begin{equation} \label{eq_grad}
\begin{split}
\nabla_{\bbtheta} \mathbb{E}_\bbh [f(\bbh, \bbPhi(\bbh,\bbtheta))] = \mathbb{E}_{\bbh,\bbr}[f(\bbh,\bbr) \nabla_{\bbtheta} \log \pi_{\bbh,\bbtheta}(\bbr)]
\end{split}
\end{equation}
where $\bbr$ is a random realization drawn from the distribution $\pi_{\bbh,\bbtheta}(\bbr)$. We now translate the computation of $\nabla_{\bbtheta} \mathbb{E}_\bbh [f(\bbh, \bbPhi(\bbh,\bbtheta))]$ to a function evaluation $f(\bbh,\bbr)$ multiplied with the gradient of the density function $ \nabla_{\bbtheta} \log \pi_{\bbh,\bbtheta}(\bbr)$. However, computing $ \nabla_{\bbtheta} \log \pi_{\bbh,\bbtheta}(\bbr)$ for a delta density function still requires the knowledge of $f(\cdot)$. We further address this issue by approximating the delta density function with a known density function centered around $\bbPhi(\bbh, \bbtheta)$, such as the Gaussian distribution, the Binomial distribution, etc. We can then estimate the gradient of policy function $\nabla_{\bbtheta} \mathbb{E}_\bbh [f(\bbh, \bbPhi(\bbh,\bbtheta))]$ by using \eqref{eq_grad}, which does not require the function model $f(\cdot)$ but rather the function value $f(\bbh,\bbr)$ at a sampled channel state $\bbh$ and the distribution $\pi_{\bbh,\bbtheta}(\bbr)$. To estimate the expectation $\mathbb{E}_{\bbh, \bbr}[\cdot]$, we observe $\ccalT$ samples of the CSI and take the average as
\begin{equation}
\begin{split} \label{eq_policy}
\!\!\!\widetilde{\nabla_{\bbtheta}} \mathbb{E}_\bbh [f(\bbh, \bbPhi(\bbh,\bbtheta))] \!=\! \frac{1}{\ccalT}\!\sum_{\tau = 1}^\ccalT\! f(\bbh_\tau,\bbr_\tau) \nabla_{\bbtheta} \log \pi_{\bbh_\tau,\bbtheta}(\bbr_\tau)
\end{split}
\end{equation}
where $\bbh_\tau$ is a sampled CSI and $\bbr_\tau$ is a realization drawn from the distribution $\pi_{\bbh_\tau,\bbtheta}(\bbr)$. With the use of \eqref{eq_policy}, we can compute the gradient of policy function in \eqref{eq_priup1} as
\begin{align} \label{eq_polgrad}
&\widetilde{\nabla_{\bbtheta}} \mathcal{L}(\bbtheta,\lambda) =\! \frac{1}{\ccalT}\!\sum_{\tau = 1}^\ccalT\! \Big\{ \Big[ f(\bbh_\tau,\! \bbr_\tau) \! - \!\sum_{s=1}^S\! \lambda_s c_s\big(\bbr_\tau,\! f(\bbh_\tau, \!\bbr_\tau) \Big] \nabla_{\bbtheta} \!\log \pi_{\bbh_\tau,\bbtheta}(\bbr_\tau) \!\Big\}.
\end{align}

We stress the model-free aspect of computing the gradient in \eqref{eq_polgrad}. That is, we need only observe the values $f(\bbh_\tau, \bbr_\tau)$ and $c_s(\bbr_\tau, f(\bbh_\tau, \bbr_\tau))$ as experienced in the FSO system under the instantaneous observed states $\bbh_\tau$ and $\bbr_\tau$. This is considered model-free because it does not require an explicit mathematical model of $f(\cdot)$ and $c_s(\cdot)$ or the CSI distribution model, as typically required to compute analytic gradients. In terms of the dual step, by estimating the expectation with the average of $\ccalT$ samples, it can also be computed with only observed values $c_s(\bbr_\tau, f(\bbh_\tau, \bbr_\tau))$. By replacing $\nabla_{\bbtheta} \mathcal{L}(\bbtheta^k,\lambda^k)$ with $\widetilde{ \nabla_{\bbtheta}} \mathcal{L}(\bbtheta^k,\lambda^k)$ in \eqref{eq_priup1}, the resulting PDDL algorithm is model-free and summarized in Algorithm \ref{alg:learning}. 

The PDDL algorithm learns the optimal resource allocation by updating the primal and dual variables without requiring any explicit knowledge of the objective function, the constraint functions or the CSI distribution, but only their observed values. Therefore, we can perform algorithm given any generic FSO systems and required constraints.

\begin{remark} \normalfont
The learning process of the PDDL algorithm outlined in Algorithm \ref{alg:learning} may take a number of iterations to update the DNN parameters before convergence. We stress, however, that the learning process is completed offline before the real-time implementation and the total training time thus does not matter. At runtime, the execution of the learned DNN $\bbPhi(\cdot, \bbtheta^*)$ on the instantaneous channel state information $\bbh$ requires little computational complexity, yielding an efficient implementation as validated in numerical experiments.
\end{remark}

%!TEX root = mainOp.tex
%%%%%%%%%%%%%%%%%%%%%%%%%%%%%%%
%%% SECTION : Simulations   %%%
%%%%%%%%%%%%%%%%%%%%%%%%%%%%%%%

\section{Numerical Experiments} \label{sec:sims}

%%%%%%%%%%%%%%%%%%%%%%%%%%%%%%%%%%%%%%%%%%%%%%%%%%%%%%%%%%%%%%%%%%%%%%%%%%%%%%%%%%%%%%%%%%%%%%%%%%%%%%%%%%%%%%%%%%%%%%%%%%%%%%%%%%%%%%%%%%%%%%%%%%%%%%%%%%%%%%%%%%%%%%%%%%%%%%%%%%%%%%%%%%%%%%%%%%%%%%%%%%%%%%%%%%%%%%%%%%%%%%%%%%%%%%%%%%%%%%%%%%%%%%%%%%%%%%%%%%%%%%%%%%%%%%%%%%%%%%%%%%%%

In this section, we corroborate theory by numerically analyzing the performance of the SDG and PDDL algorithms for a large set of resource allocation problems in FSO communications. To implement the algorithms, we consider a batch-size of $\ccalT = 64$ samples. The ADAM optimizer is used for the primal update and the exponentially decaying step-size is used for the dual update. In the PDDL algorithm, we address the feasibility condition $\bbr(\bbh) \in \ccalR$ or $\bbtheta \in \Theta$ by selecting suitable policy distributions $\pi_{\bbh,\bbtheta}$, as detailed in specific applications. Also note that though the PDDL algorithm is model-free, we make system observations (objective and constraint function observations) in numerical simulation using a given model; however, we do not assume knowledge of this model to implement the algorithm, only to generate samples to be observed.

\noindent\textbf{Channel state information.} We consider FSO channel effects comprising two components: the attenuation $h_a$ and the turbulence $h_t$. The attenuation $h_a$ represents the path loss induced by weather conditions as $h_a = A_t A_r e^{-\alpha d}/(d^2 \lambda^2)$ with $\alpha$ the attenuation coefficient depending on weather visibility, $d$ the transmission distance, $\lambda$ the wavelength, $A_t$ and $A_r$ the aperture areas of the transmitter and the receiver. The turbulence $h_t$ is modeled as the well-known log-normal distribution, which is commonly used under weak-to-moderate turbulence. Without loss of generality, other distributions (e.g., Gamma-Gamma distribution) are applicable based on turbulence conditions. We then characterize the FSO channel as $y = h_a h_t x +n$ with $x$ the transmitted signal, $y$ the received signal and $n$ the additive Gaussian noise.

\subsection{Power Adaptation}\label{exppower}

We first consider the power adaptation in the RoFSO system---see Section \ref{poweradaption}. The goal is to allocate powers to orthogonal optical carriers that maximize the weighted sum-capacity within total and peak power constraints
\begin{alignat}{3} \label{eq_problem1112}
 \mathbb{P}:= &  \max_{\bbr(\bbh)} \ && \sum_{i=1}^N \omega_i \mathbb{E}_\bbh\left[C_i(\bbh, \bbr(\bbh))\right] ,             \\
        &  \st                    \ && \mathbb{E}_\bbh\Big[\sum_{i=1}^N p_i(\bbh)\Big] - P_t\le 0,~ \ccalR &= [0,P_s]^N   \nonumber %
\end{alignat}
with $C_i(\bbh, \bbr(\bbh))$ the capacity of $i$-th wavelength channel [cf. \eqref{eq_capacity}], $\bbr(\bbh)=[p_1(\bbh),\ldots,p_N(\bbh)]^\top$ the allocated powers, $P_t$ and $P_s$ the total and peak power limitations. The problem is challenging due to the complicated non-convex objective and constraints.

\begin{figure*}%
\centering
\begin{subfigure}{0.33\columnwidth}
\includegraphics[width=1.0\linewidth, height = 0.7\linewidth]{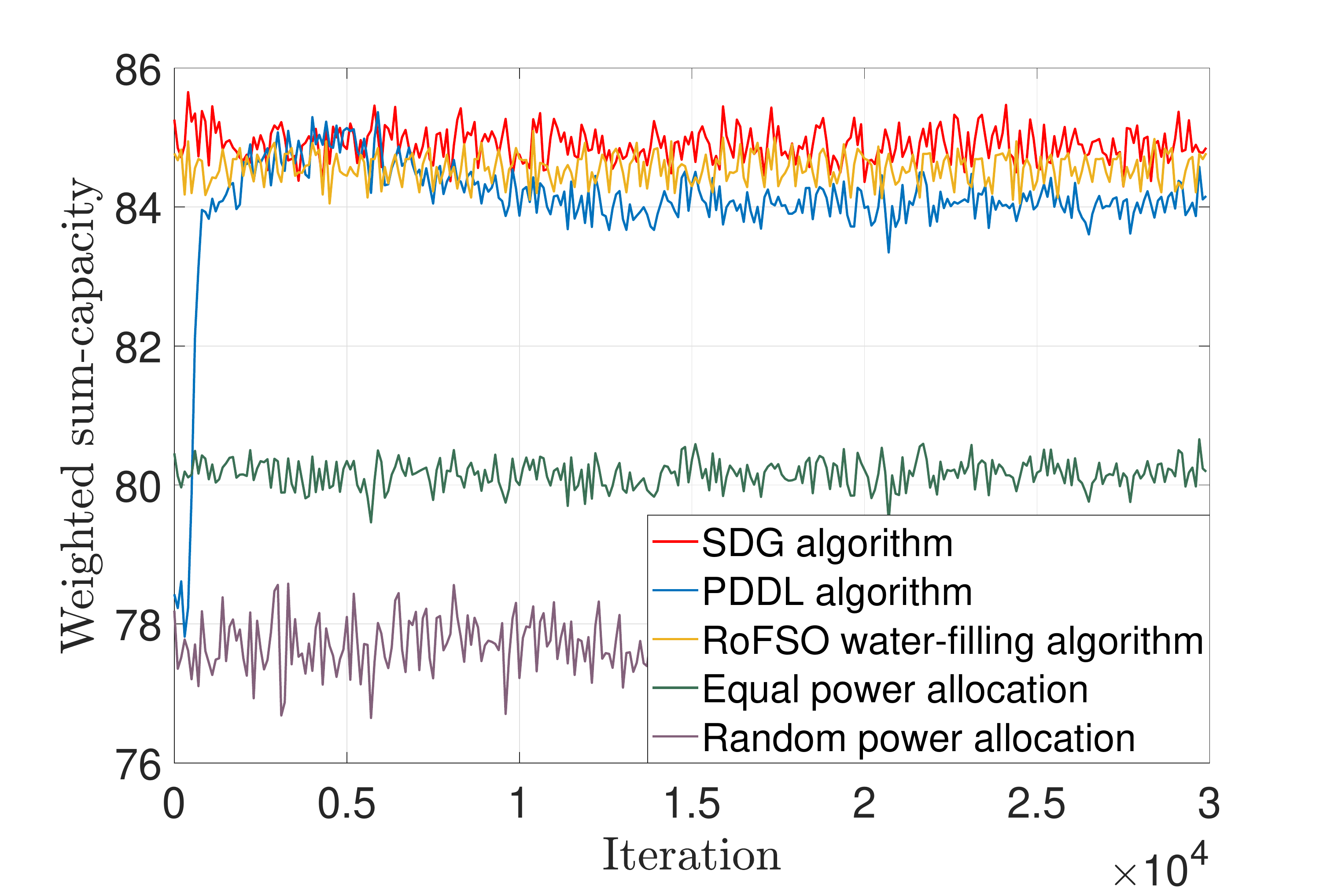}%
\caption{}%
\label{subfiga_obj_normal}%
\end{subfigure}\hfill\hfill%
\begin{subfigure}{0.33\columnwidth}
\includegraphics[width=1.0\linewidth,height = 0.7\linewidth]{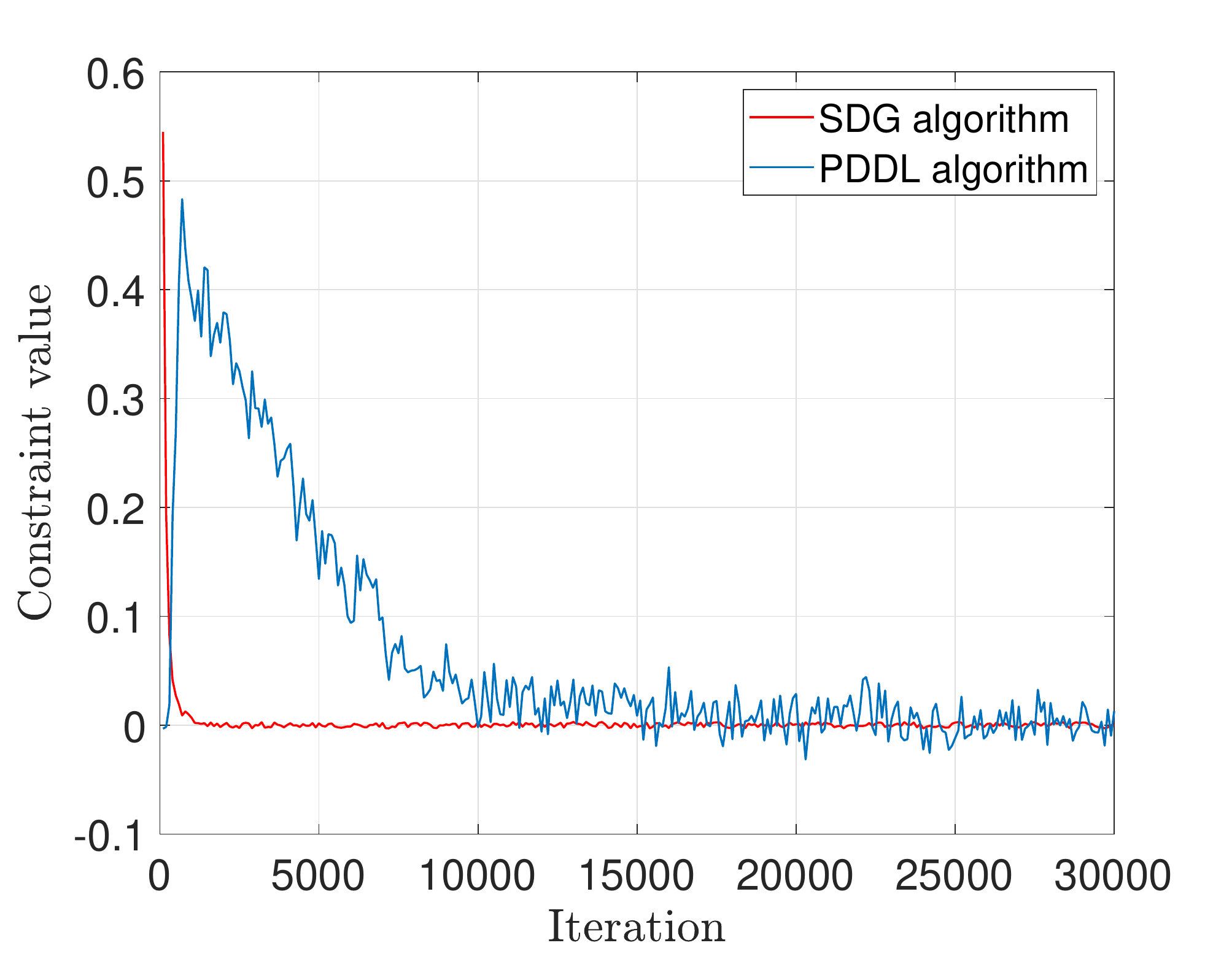}%
\caption{}%
\label{subfigb_con}%
\end{subfigure}\hfill\hfill%
\begin{subfigure}{0.33\columnwidth}
\includegraphics[width=1.0\linewidth,height = 0.7\linewidth]{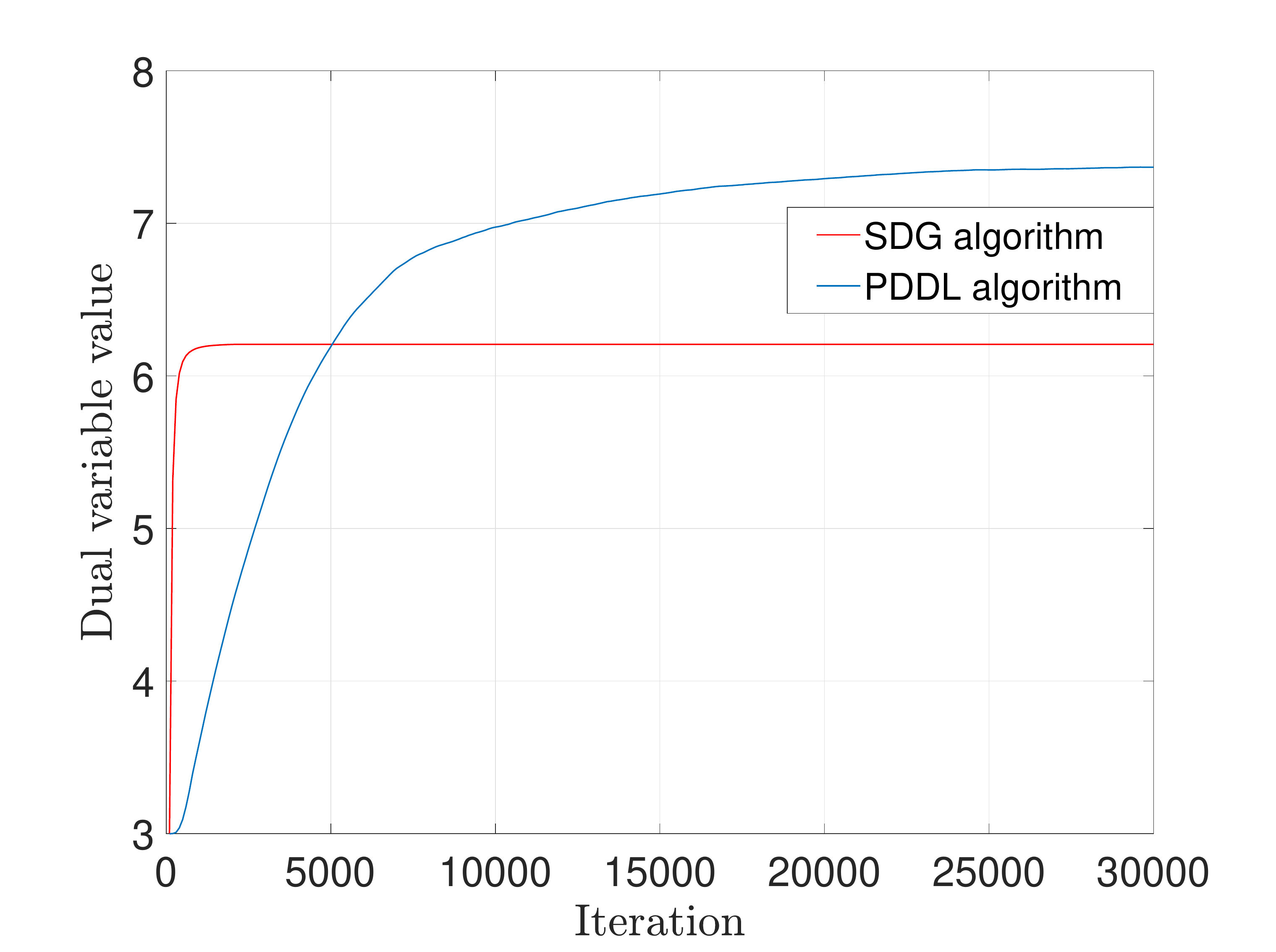}%
\caption{}%
\label{subfigc_lam}%
\end{subfigure}
\caption{Performance of the SDG, the PDDL and the baseline policies for power adaptation in the $10$ wavelength multiplexing RoFSO system. (a) The objective value. (b) The constraint value. (c) The dual variable value.}\label{fig_simple_power}\vspace{-5mm}
\end{figure*}

\begin{figure}%
\centering
\begin{subfigure}{0.33\columnwidth}\centering
\includegraphics[width=1\linewidth, height = 0.7\linewidth]{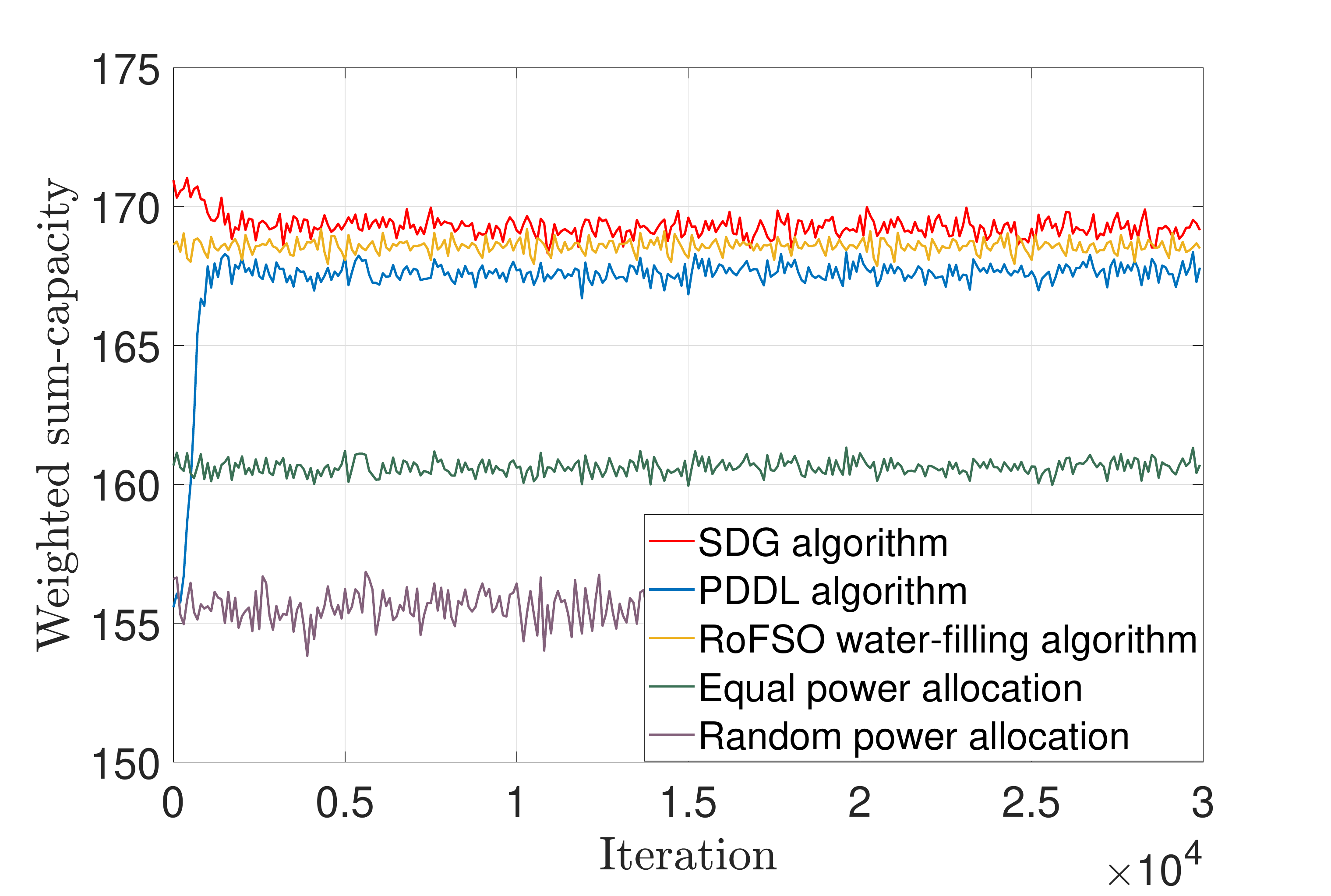}%
\caption{}%
\label{subfiga_largeobj}%
\end{subfigure}\hfill\hfill%
\begin{subfigure}{0.33\columnwidth}\centering
\includegraphics[width=1\linewidth, height = 0.7\linewidth]{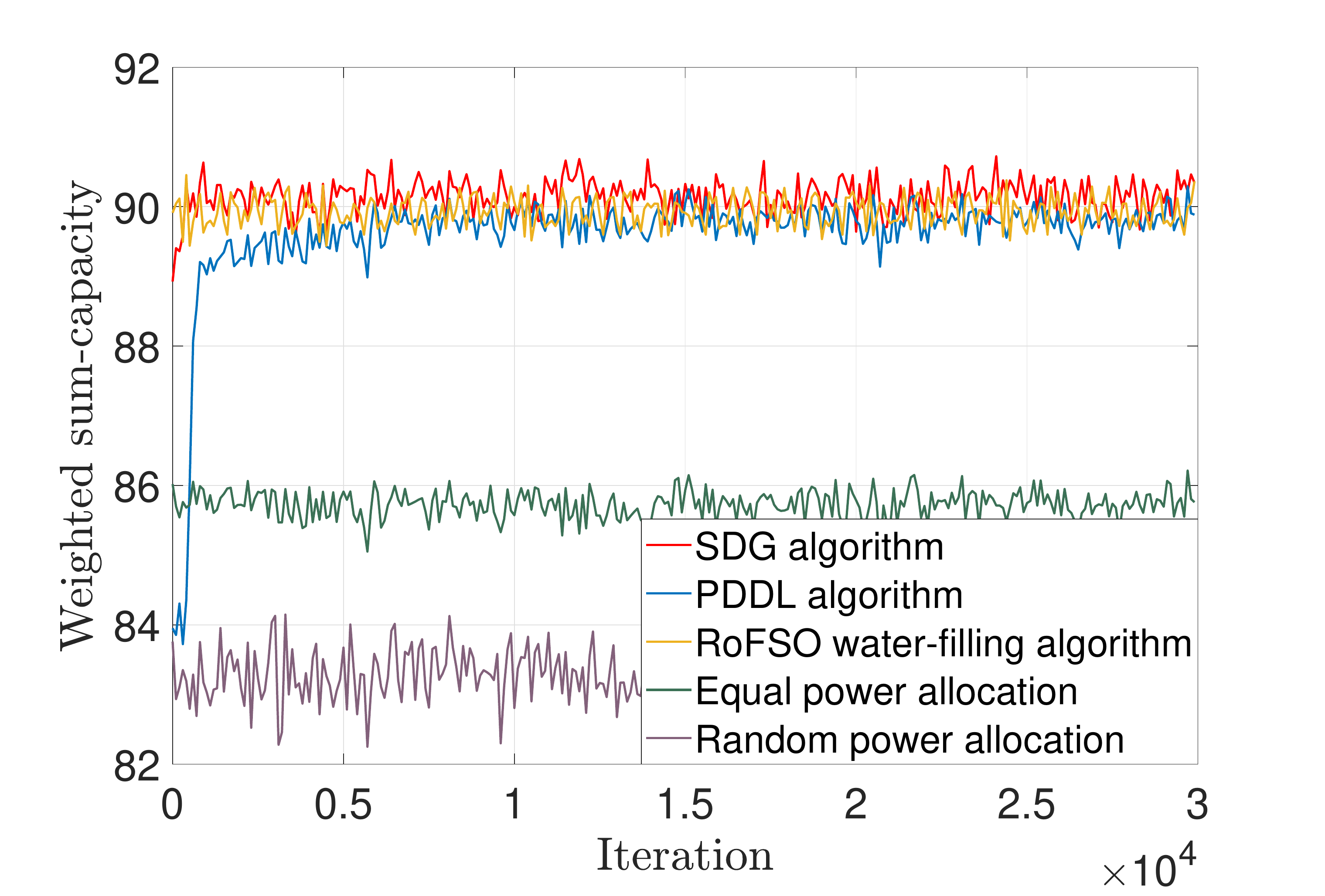}%
\caption{}%
\label{subfiga_moreobj}%
\end{subfigure}\hfill\hfill%
\begin{subfigure}{0.33\columnwidth}
\includegraphics[width=1.0\linewidth,height = 0.7\linewidth]{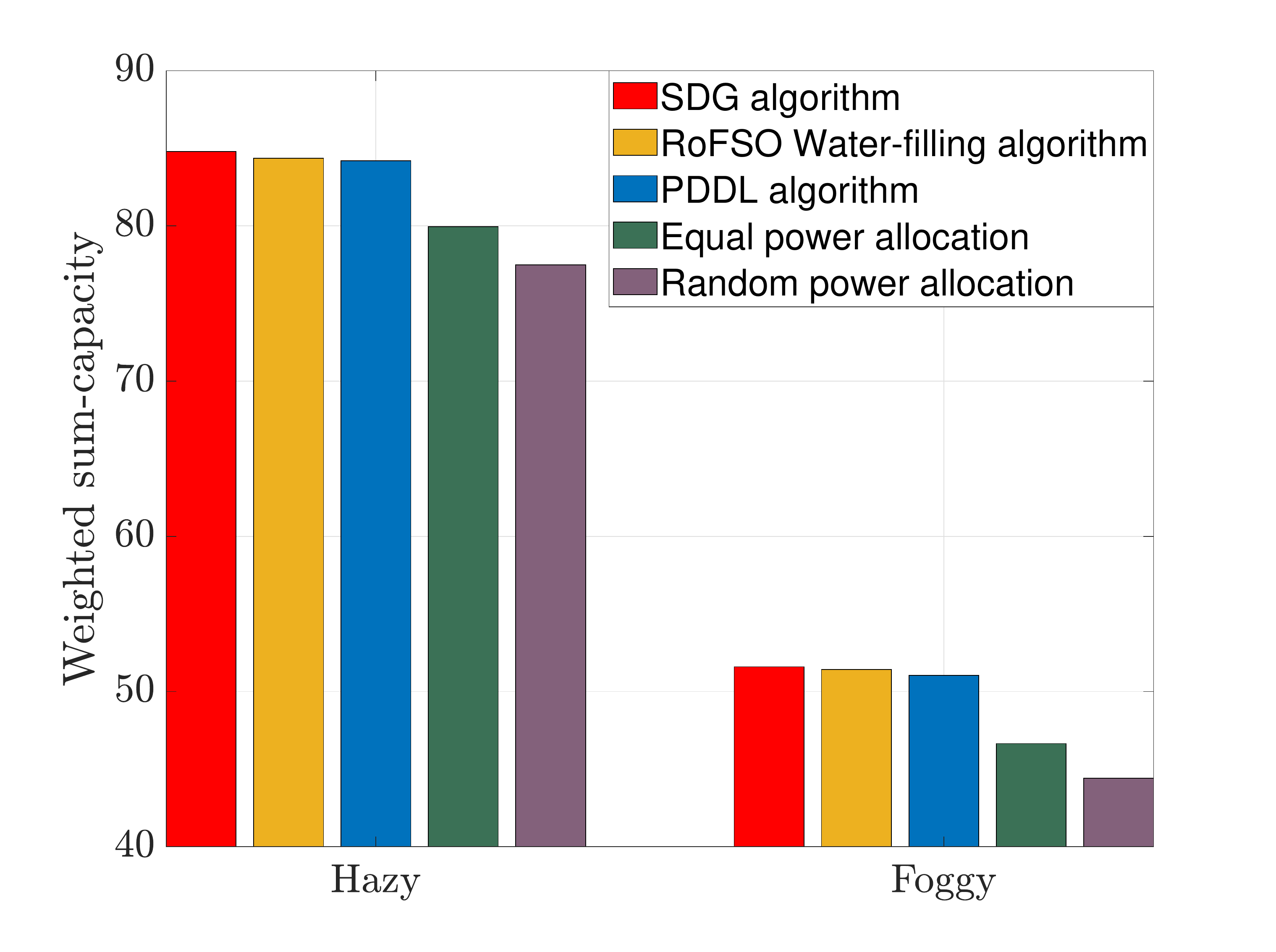}%
\caption{}%
\label{subfigd_lam}%
\end{subfigure}%
\caption{Performance of the SDG, the PDDL and the baseline policies for power adaptation in different RoFSO system configurations. (a) $20$ wavelength multiplexing with power limitations $P_t = 3W, P_s = 0.3W$. (b) $10$ wavelength multiplexing with power limitations $P_t = 3W, P_s = 0.6W$. (c) Hazy and light foggy weather conditions.}\label{fig_complex_power}\vspace{-5mm}
\end{figure}

The priority weights $\bbomega$ are drawn randomly in $[0,1]$ and system parameters are set as: $P_t = 1.5W$; $P_s = 0.3W$; $m_p=5$; $OMI=15\%$; $r=0.75$; $RIN=-140dB/Hz$; $T=300K$; transmitter aperture diameter $D_{tx}=0.015m$; receiver aperture diameter $D_{rx} = 0.05m$ and $d=1km$. We consider three baseline policies for comparison: (i) the RoFSO water-filling algorithm (ii) the average power allocation and (iii) the random power allocation. The first is the modified (improved) water-filling algorithm for the RoFSO system depending on system models to solve KKT conditions \cite{zhou2015optical}, while the second and the third are model-free. For the PDDL algorithm, we consider the policy distribution $\pi_{\bbh,\bbtheta}$ as a truncated Gaussian distribution to satisfy the feasibility condition $\bbr(\bbh)\in \ccalR = [0,P_s]^N$, i.e., the truncated Gaussian distribution has fixed support on $[0,P_s]$. Since there is no coupling or interference between wavelength channels, we construct $N$ independent DNNs serving for $N$ channels. The input of each DNN is the CSI on its associated channel, and the output $\bbPhi(\bbh,\bbtheta)\in \mathbb{R}^{2}$ is a set of mean and standard deviation that specify the truncated Gaussian distribution. The DNN is built with two hidden layers, each containing $20$ and $10$ units respectively, and the nonlinearity is the ReLU $\sigma(\cdot) = [\cdot]_{+}$.

Fig. \ref{fig_simple_power} shows results of a relatively small-scale experiment with $N=10$ wavelength channels. From Fig. \ref{subfiga_obj_normal}, we see that the SDG and the PDDL converge as the iteration increases. The SDG solves the problem exactly and thus exhibits the best performance than baseline policies. The PDDL outperforms significantly the other model-free policies and achieves close performance to that of the model-based SDG and water-filling algorithms, indicating its near-optimal performance without model knowledge. Fig. \ref{subfigb_con} shows that the constraint value converges to zero with the increase of iteration. This confirms the feasibility of solutions obtained by our algorithms. In Fig. \ref{subfigc_lam}, we observe that the dual variable learned by the PDDL converges closely to that of the SDG, with a small difference as predicted by the near-universality of DNNs. 

In Fig. \ref{fig_complex_power}, we run experiments under different system configurations; namely, different number of wavelength channels, different power budgets, and different weather conditions, to show the algorithm adaptability to changing scenarios. Fig. \ref{subfiga_largeobj} plots the objective in the RoFSO system with $N=20$ wavelength multiplexing, Fig. \ref{subfiga_moreobj} shows that with larger power budgets $P_t = 3W$ and $P_s=0.6W$, and Fig. \ref{subfigd_lam} compares the hazy (4.5dB/km path loss exponent) and light foggy (11.5dB/km path loss exponent) weather conditions. Similar results apply here, where the SDG outperforms baseline policies and the PDDL achieves near-optimal performance in a model-free manner. We also observe that performance improvements of the SDG and the PDDL compared to the model-free baseline policies become more visible in larger systems (Fig. \ref{subfiga_largeobj}) and worse weather conditions (Fig. \ref{subfigd_lam}), and the PDDL converges roughly to the same value as the SDG with larger power budgets (Fig. \ref{subfiga_moreobj}). The latter is because the increased budgets create more space for the PDDL to manipulate powers, such that the learning ability of DNNs is fully activated.
%%%%%%%%%%%%%%%%%%%%%%%%%%%%%%%%%%%%%%%%%%%%%%%%%%%%%%%%%%%%%%%%%%%%%%%%%%%%%%%%%%%%%%%%%%%%%%%%%%%%%%%%%%%%%%%%%%%%%%%%%%%%%%%%%%%%%%%%%%%%%%%%%%%%%%%%%%%%%%%%%%%%%%%%%%%%%%%%%%%%%%%%%%%%%%%%%%%%%%%%%%%%%%%%%%%%%%%%%%%%%%%%%%%%%%%%%%%%%%%%%%%%%%%%%%%%%%%%%%%%%%%%%%%%%%%%%%%%%%%%%%%%
\begin{table}[t] 
\begin{center}  
\caption{Implementation time required for the SDG, the PDDL and the water-filling algorithms in three cases. (a) $10$ wavelength multiplexing with $P_t\!=\!1.5$W, $P_s\!=\!0.3$W. (b) $10$ wavelength multiplexing with $P_t\!=\!3$W, $P_s \!=\! 0.6$W. (c) $20$ wavelength multiplexing with $P_t\!=\!3$W, $P_s \!=\! 0.3$W.}  
\label{table1}
\begin{tabular}{|l|l|l|l| p{2cm}|}  
\hline  
 & Case (a)  & Case (b) & Case (c) \\ \hline  
The SDG &  $2.81 \cdot 10^{-3}$s & $3.91 \cdot 10^{-3}$s & $7.19 \cdot 10^{-3}$s \\ \hline  
The PDDL &  $1.56 \cdot 10^{-5}$s & $1.52 \cdot 10^{-5}$s & $1.59\cdot 10^{-5}$s\\  \hline
The RoFSO water-filling & 1.65s & 1.71s & 3.31s\\ 
\hline  
\end{tabular}  
\end{center}  \vspace{-4mm}
\end{table}

Besides performance, the implementation time is of utmost importance for cooperative transmissions that allocate resources based on instantaneous CSI. Table \ref{table1} compares the implementation time of the SDG, the PDDL and the water-filling algorithms for processing an instantiation of CSI. We see that the PDDL requires far less time than the other algorithms but achieves comparable performance. This is because the computation of the DNN contains simply linear operations with pointwise nonlinearities, whereas the SDG requires some more computation expense for solving the inner maximization in \eqref{eq_priup}. The water-filling algorithm is particularly computationally expensive, requiring substantial time to solve the KKT conditions of the complicated objective with power constraints. The time saved by the PDDL and the SDG increases as the system becomes larger, highlighting a further advantage of our algorithms.

\subsection{Relay Selection}\label{exp:relay}

\begin{figure*}%
\centering
\begin{subfigure}{0.25\columnwidth}
\includegraphics[width=1.0\linewidth, height = 0.7\linewidth]{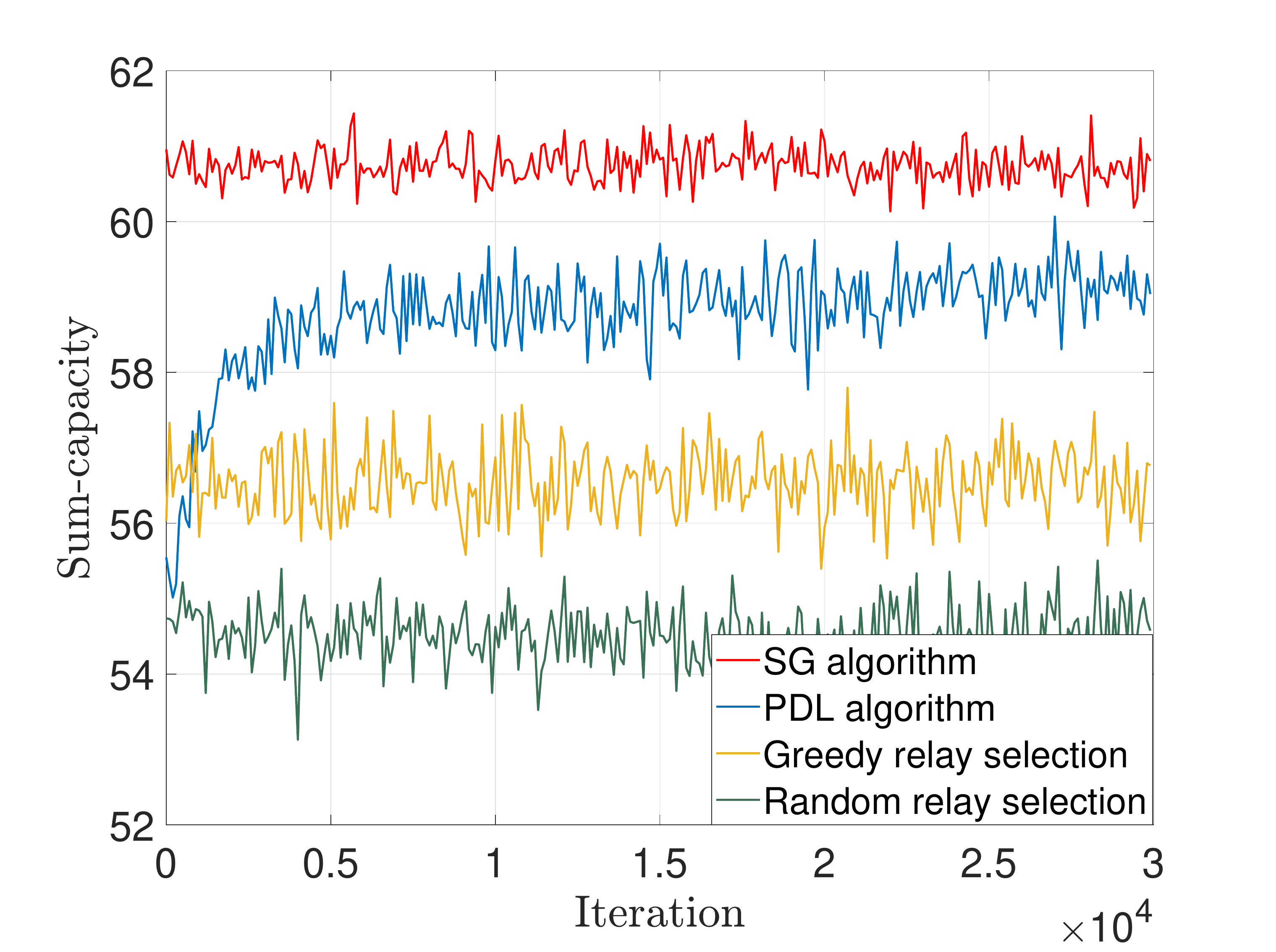}%
\caption{}%
\label{subfiga_relay}%
\end{subfigure}\hfill\hfill%
\begin{subfigure}{0.25\columnwidth}
\includegraphics[width=1.0\linewidth,height = 0.7\linewidth]{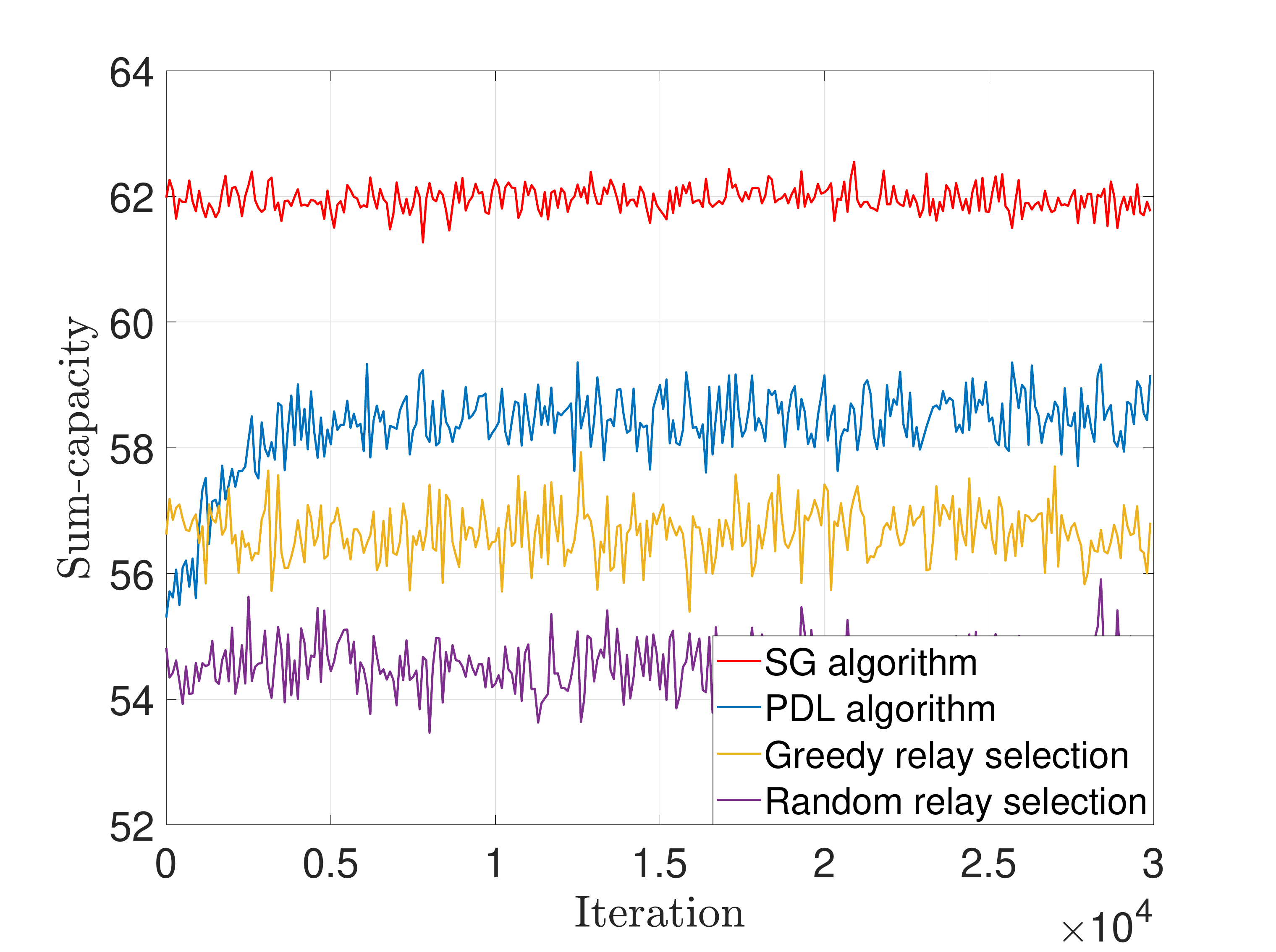}%
\caption{}%
\label{subfigb_relay}%
\end{subfigure}\hfill\hfill%
\begin{subfigure}{0.25\columnwidth}
\includegraphics[width=1.0\linewidth,height = 0.7\linewidth]{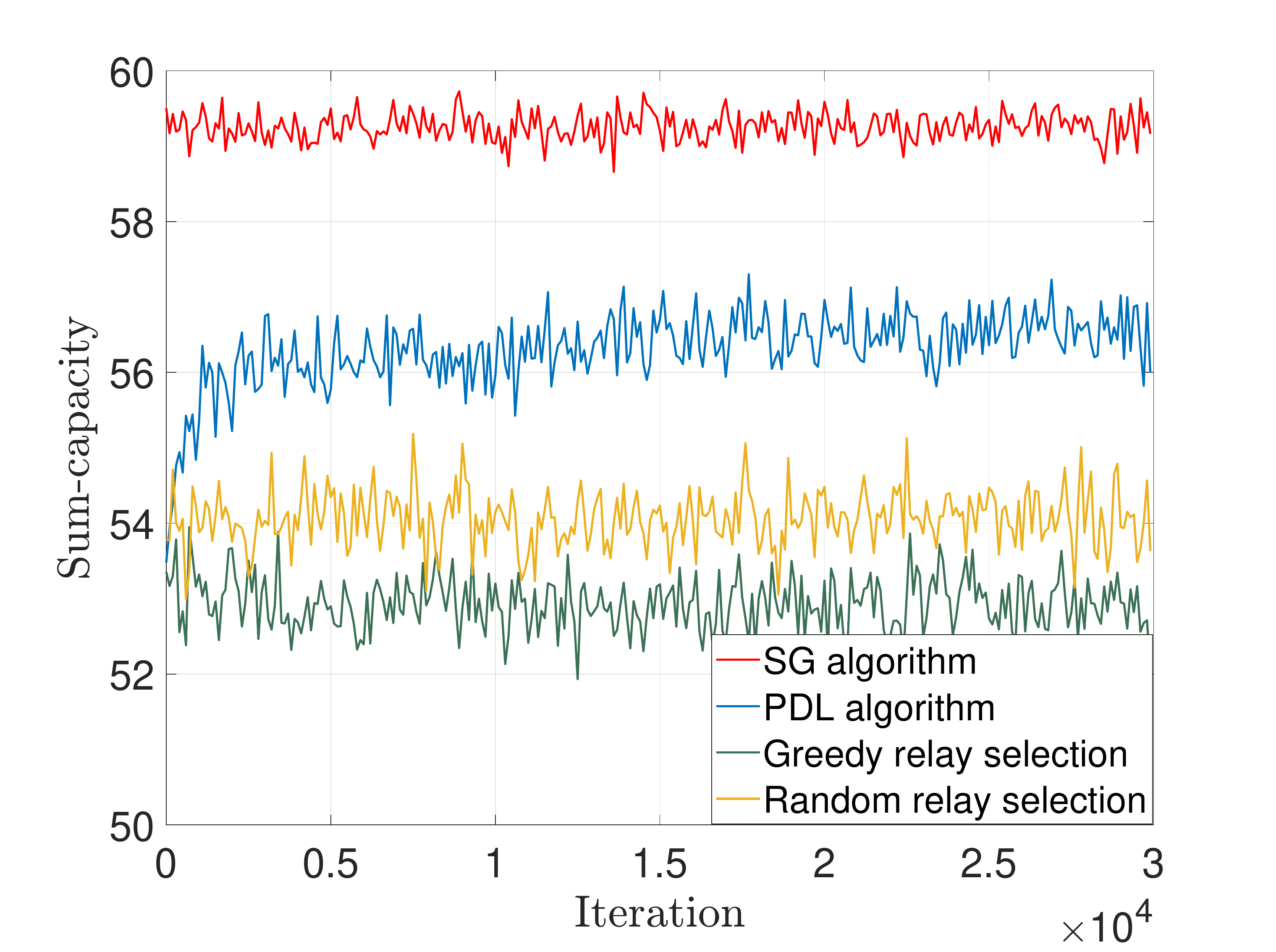}%
\caption{}%
\label{subfigc_relay}%
\end{subfigure}\hfill\hfill%
\begin{subfigure}{0.25\columnwidth}
\includegraphics[width=1.0\linewidth,height = 0.7\linewidth]{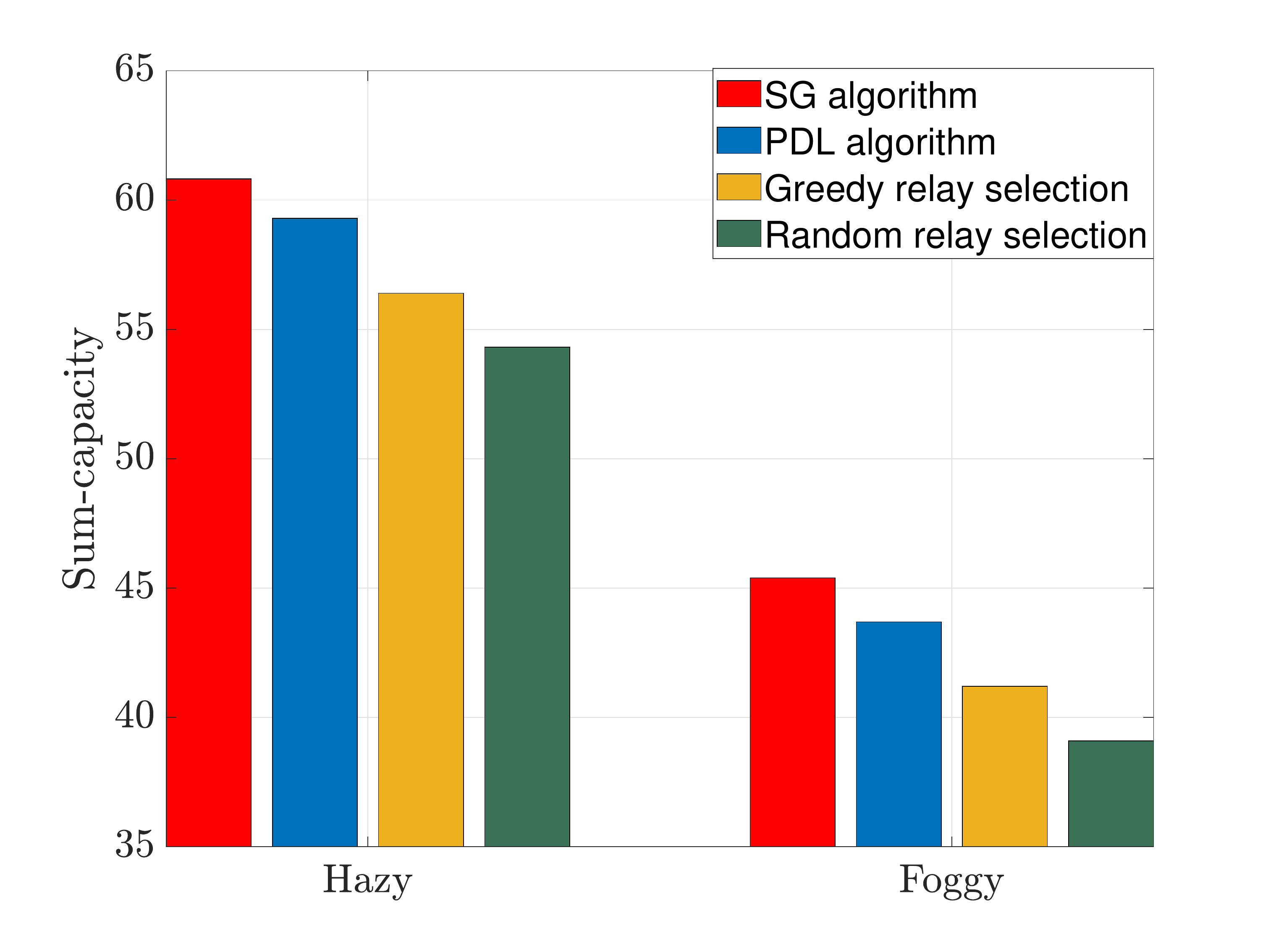}%
\caption{}%
\label{subfigd_relay}%
\end{subfigure}%
\caption{Performance of the SG, the PDL and the baseline policies in relay-assisted FSO networks. (a) The $2$-hop FSO network with $5$ parallel relays per hop. (b) The $2$-hop FSO network with $10$ parallel relays per hop. (c) The $3$-hop FSO network with $5$ parallel relays per hop. (d) Hazy and light foggy weather conditions.}\label{fig_relay}\vspace{-5mm}
\end{figure*}

We then consider the relay selection in relay-assisted FSO networks---see Section \ref{relayselection}. The goal is to select the appropriate relay at each hop to maximize the channel capacity
 \begin{alignat}{3} \label{eq_problem_relay}
 \mathbb{P}:= &  \max_{\bbr(\bbh)}\! \ && \mathbb{E}_\bbh \Big[ \sum_{j_N = 1}^M \!\cdots\! \sum_{j_1 = 1}^M \left(\prod_{i=1}^N \alpha_{i j_i}(\bbh) \right) C_{j_1 \ldots j_N} (\bbh)\Big],             \\
        &  \st           \!         \ && \!\ccalR \!=\! \Big\{\! \{0,1\}^{N\!\times\! M} | \sum_{j_i=1}^M \!\alpha_{ij_i}(\bbh) \!\leq\! 1,\forall~i\!=\!1,...,\!N \!\Big\}.  \nonumber  %
\end{alignat}
where $C_{j_1 \ldots j_N} (\bbh)$ is the channel capacity of the relaying link [cf. \eqref{eq_capacity35}], and $\bbr(\bbh) = [\bbalpha_1(\bbh), \ldots, \bbalpha_N(\bbh)]^\top \in \{ 0, 1 \}^{N \times M}$ are selected relays. Note that there is no stochastic constraint in problem \eqref{eq_problem_relay}, i.e., there is no constraint taking the form of $\mathbb{E}_\bbh\! \left[c_s\big(\bbr(\bbh), f(\bbh, \bbr(\bbh))\big)\right]\le 0$, in which case the dual update is not required. The SDG algorithm reduces to the Stochastic Gradient (SG) algorithm and the PDDL algorithm reduces to the Primal Deep Learning (PDL) algorithm, respectively.

The system parameters are set as $B=5\times 10^8 Hz$, $T_f = 10^{-8}s$, $\epsilon=1$, $P=0.3W$ and $R=0.75 A/W$. We consider two baseline policies for performance comparison: (i) the greedy policy that selects the relay with lowest CSI at each hop and (ii) the random relay selection, where both are model-free. For the PDL algorithm, we select the policy distribution $\pi_{\bbh,\bbtheta}$ as the categorical distribution since the allocated resources are binary variables $\bbr(\bbh) \in \ccalR $. The categorical distribution describes a random variable that takes on one of $M$ possible categories. We construct a single two-layered DNN of $200$ and $100$ hidden units and the nonlinearity is the ReLU. Channel conditions over the network $\bbh$ are given as inputs to the DNN, which outputs $\bbPhi(\bbh, \bbtheta) \in \mathbb{R}^{NM}$ that specify the selected probabilities of each relay (category) at each hop.

\begin{table}[t] 
\begin{center}  
\caption{Implementation time required for the SG, the PDL and the greedy relay selection in three cases. (a) $2$-hop FSO network with $5$ parallel relays per hop. (b) $2$-hop FSO network with $10$ parallel relays per hop. (c) $3$-hop FSO network with $5$ parallel relays per hop.}  
\label{table2}
\begin{tabular}{|l|l|l|l| p{2cm}|}  
\hline  
 & Case (a)  & Case (b) & Case (c) \\ \hline  
The SG &  $1.40\cdot 10^{-4}$s & $5.77\cdot 10^{-4}$s & $9.81\cdot 10^{-4}$s \\ \hline  
The PDL &  $1.55\cdot 10^{-5}$s & $3.09 \cdot 10^{-5}$s & $3.11 \cdot 10^{-5}$s\\  \hline
The greedy & $1.49 \cdot 10^{-5}$s & $1.56 \cdot 10^{-5}$s & $1.55 \cdot 10^{-5}$s\\ 
\hline  
\end{tabular}  
\end{center}  \vspace{-4mm}
\end{table}

Fig. \ref{subfiga_relay} exhibits the performance of the SG, the PDL and two baseline policies in a $2$-hop network with $N=5$ parallel relays per hop. We see that both the SG and the PDL converge and outperform the baseline policies. The SG performs best on the premise that system models are available at hand. The PDL follows closely with a similar objective value, which is obtained without explicit model information. The constraints of relay selection are automatically satisfied by using the categorical distribution, confirming the feasibility of obtained solutions. Additional experiments in alternative scenarios are performed, i.e., a $2$-hop network with $N=10$ relays per hop in Fig. \ref{subfigb_relay}, a $3$-hop network with $N=5$ relays per hop in Fig. \ref{subfigc_relay}, and the hazy and light foggy weather conditions in Fig. \ref{subfigd_relay}. We observe similar results indicating the adaptivity of both algorithms to larger FSO networks and different weather conditions. The PDL gets slightly degraded in Fig. \ref{subfigb_relay} and \ref{subfigc_relay} because the problem becomes more difficult as we enlarge the system with more relays or more hops, while the DNN remains same with unchanged representational power. 

Table \ref{table2} shows the implementation time of the SG, the PDL and the greedy algorithms. Though the SG exhibits the best performance, its implementation takes the most time. The PDL achieves the close performance to the SG but only requires a comparable time as the greedy policy, achieving a favorable balance between these two factors.

\subsection{Joint Power and Relay Allocation} \label{Subsec:joint}

We now consider the joint power and relay allocation in two applications, which are more complicated but also of more interests in practice.

\textbf{Relay-assisted multichannel FSO network.} For the first experiment, we consider the relay-assisted multichannel FSO network where the system transmits signals with $L$ orthogonal optical carriers through $N$ intermediate hops \cite{hassan2016statistical}. In particular, the transmitter modulates signals onto multiple optical carriers and sends them simultaneously to the selected relay. The latter aggregates received signals, modulates orthogonal carriers, and transmits to the selected relay at next hop until the receiver. We assume there is no crosstalk between orthogonal carriers and each hop contains $M$ parallel relays for selection. Based on the CSI, different relays are selected at different hops and different powers are assigned to different carriers at the transmitter and selected relays to maximize the total channel capacity. Let $\bbh$ be the CSI between the transmitter, relays and the receiver, and $\bbr(\bbh)=\{ \bbp_{ij}(\bbh), \alpha_{ij}(\bbh)\}_{i=0,...,N,j=1,...,M}$ the allocated resources including assigned powers and selected relays. In particular, $\bbp_{ij}(\bbh) = [p_{ij}^1(\bbh),\ldots,p_{ij}^L(\bbh)]^\top \in \mathbb{R}^L$ are powers of $L$ optical carriers at $j$-th relay of $i$-th hop where $i=0, j=1$ and $i=N+1, j=1$ represent the transmitter and the receiver, and $\alpha_{ij} \in \{ 0,1 \}$ indicates whether $j$-th relay is selected at $i$-th hop. The channel capacity of $\ell$-th orthogonal channel over a specific selected relaying link is
\begin{align} \label{eq_capacity35}
&C_{j_1 \ldots j_N}^\ell (\bbh) =\frac{T_f B}{\epsilon} \log \!\Big( \!1\!+\! \Big( \prod_{i=0}^N \Big( 1+\frac{1}{p_{ij_i}^\ell(\bbh) h^{\ell}_{j_{i}j_{i+1}} \frac{R}{e \Delta f}} \Big)-1 \Big)^{-1} \Big)
\end{align}
where we assume $j_i$-th relay is selected at $i$-th hop and $h^\ell_{j_{i}j_{i+1}}$ is the CSI of $\ell$-th optical carrier between $j_{i}$-th relay at $i$-th hop and $j_{i+1}$-th relay at $(i+1)$-th hop. Since there is single transmitter and single receiver, we have $j_0=j_{N+1}=1$ by default. There are three types of constraints: the total power limitation $P_t$ at the transmitter and selected relays, the peak power limitation $P_s$ for each carrier, and that only one relay is selected at each hop. The optimization problem is
 \begin{alignat}{3} \label{eq_problem_powerrelay}
 &\mathbb{P}\!:=\! &&  \max_{\bbr(\bbh)}\! \ \mathbb{E}_\bbh\!\Big[\! \sum_{j_N \!=\! 1}^M \!\!\cdots\!\! \sum_{j_1 \!=\! 1}^M \!\Big(\prod_{i=1}^N \!\alpha_{i j_i}(\bbh)\! \Big)\! \sum_{\ell=1}^L \!\omega_\ell C_{j_1 \ldots j_N}^\ell (\bbh)\Big]\! ,              \\
        &  \st                    \ && \!\mathbb{E}_\bbh\!\Big[\!\sum_{\ell=1}^L p_{ij_i}^\ell(\bbh)\!\Big] \!-\! P_t \!\le\! 0,~i\! =\!1,\!\ldots\!,\!N, j_i\!=\!1,\!\ldots\!,\!M,\nonumber \\
        &  &&  \!\ccalR\! =\! \Big\{\! [0,\!P_s]^{(1+N\times M)\times L} \!\times\! \{ 0,1 \}^{N\!\times\! M} | \!\sum_{j_i=1}^M\! \alpha_{ij_i}(\bbh) \leq\! 1,i \!=\!1,\!...\!,\!N \Big\}\! \nonumber %
\end{alignat}
where $\bbomega = [\omega_1,\ldots,\omega_L]^\top$ represent priorities of different optical carriers. This challenging problem can be considered as the extension of the problem in Section \ref{relayselection} to the scenario with orthogonal optical carriers.

\begin{figure*}%
\centering
\begin{subfigure}{0.25\columnwidth}
\includegraphics[width=1.0\linewidth, height = 0.7\linewidth]{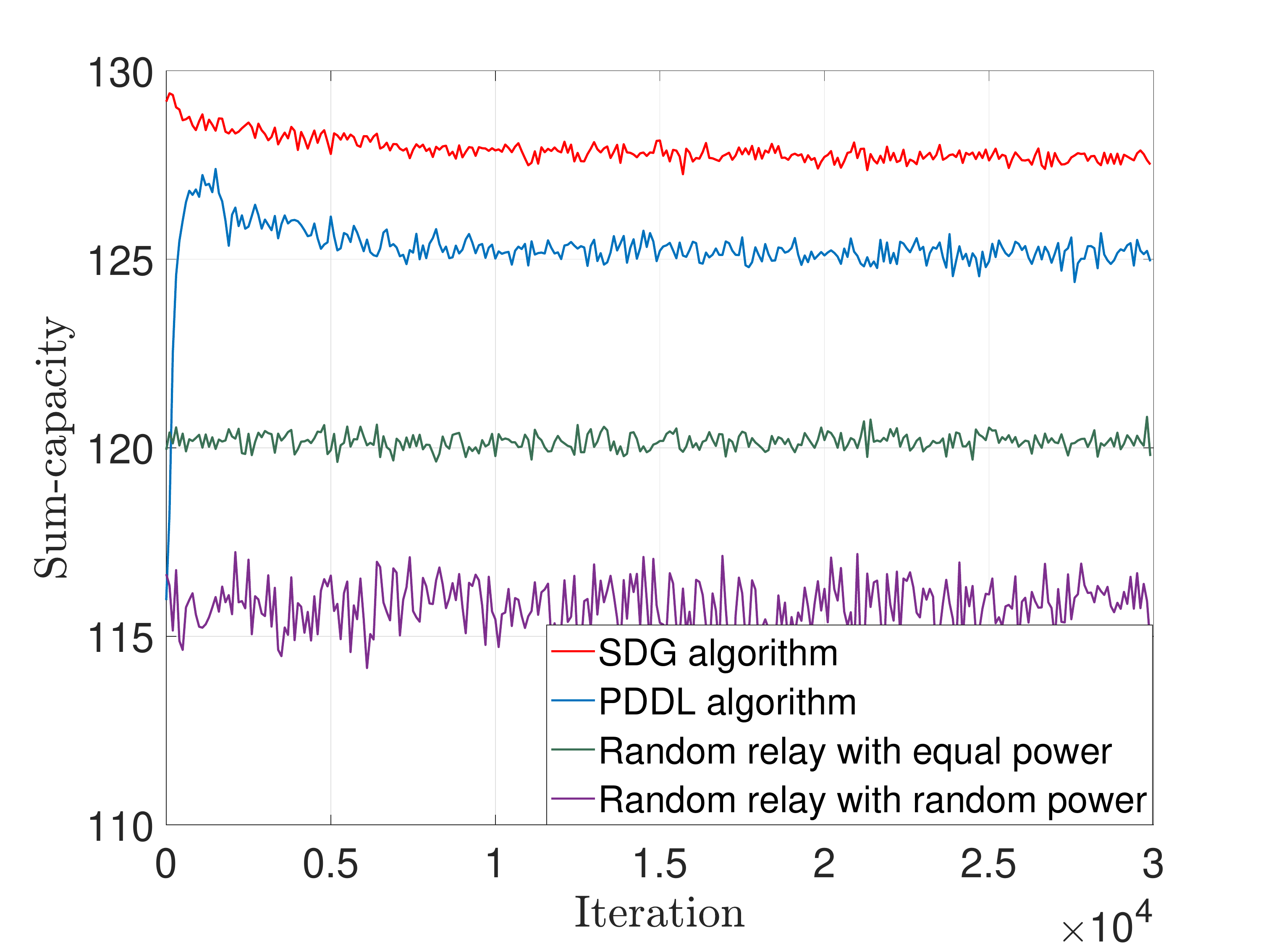}%
\caption{}%
\label{subfiga_pr}%
\end{subfigure}\hfill\hfill%
\begin{subfigure}{0.25\columnwidth}
\includegraphics[width=1.0\linewidth,height = 0.7\linewidth]{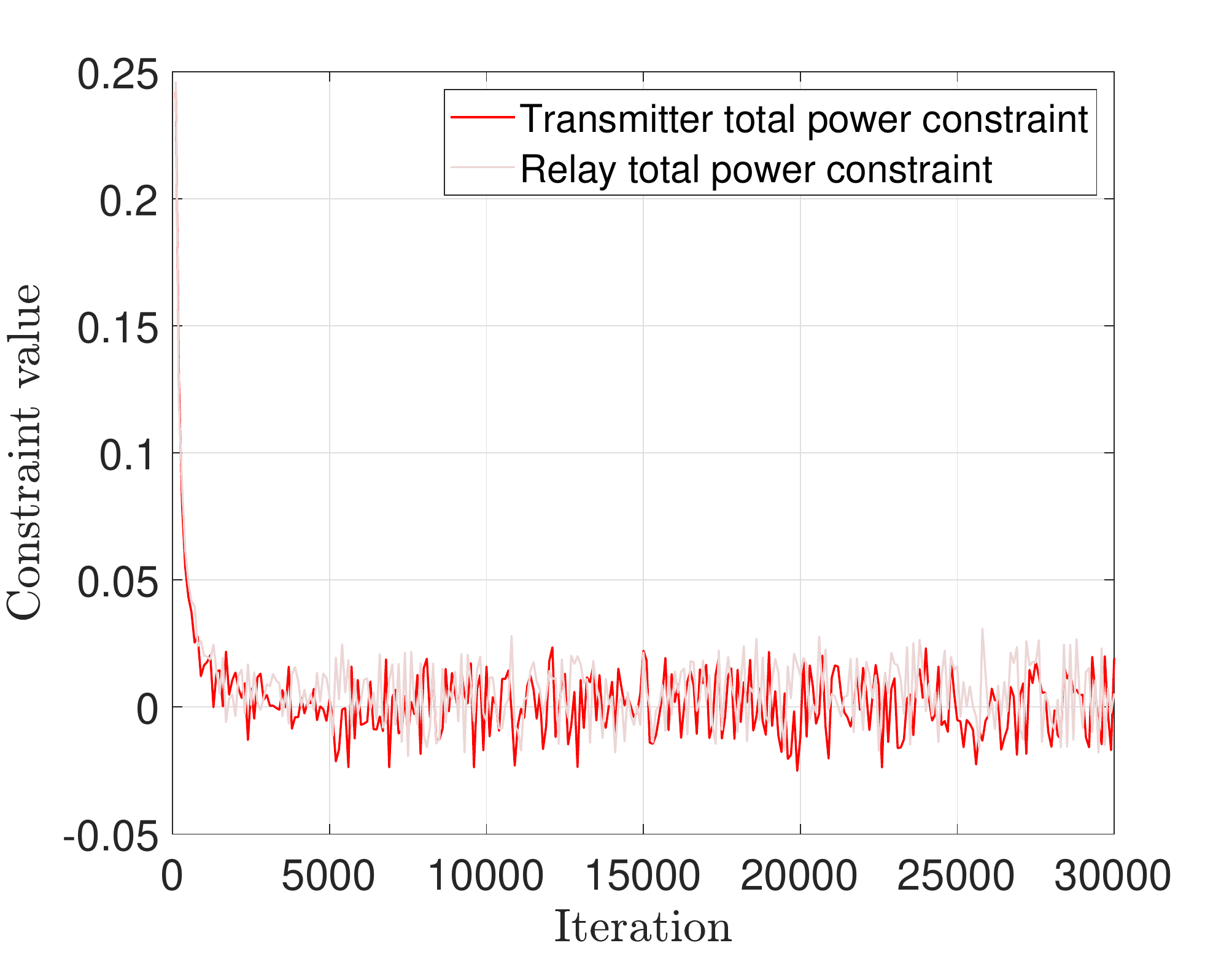}%
\caption{}%
\label{subfigb_pr}%
\end{subfigure}\hfill\hfill%
\begin{subfigure}{0.25\columnwidth}
\includegraphics[width=1.0\linewidth,height = 0.7\linewidth]{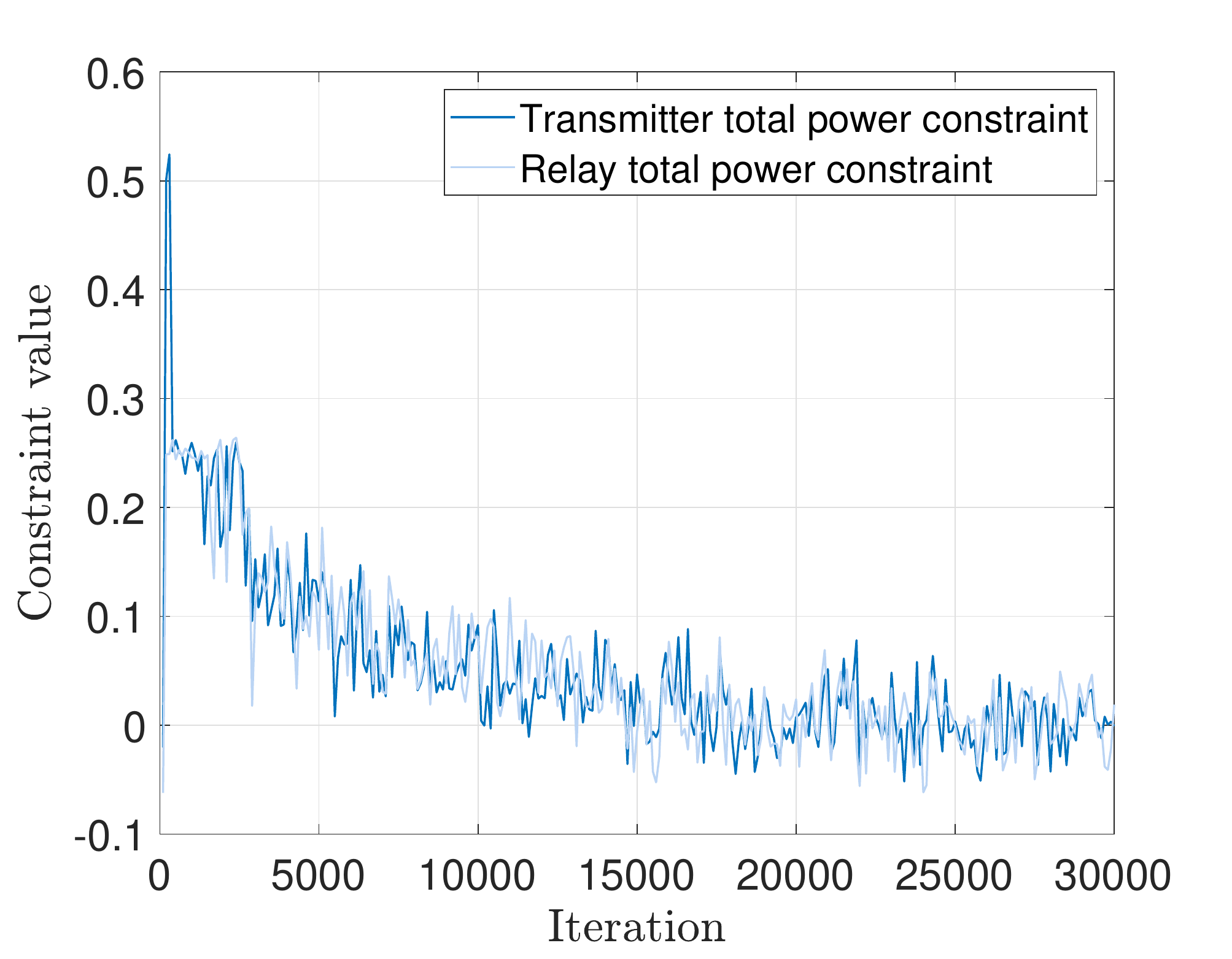}%
\caption{}%
\label{subfigc_pr}%
\end{subfigure}\hfill\hfill%
\begin{subfigure}{0.25\columnwidth}
\includegraphics[width=1.0\linewidth,height = 0.7\linewidth]{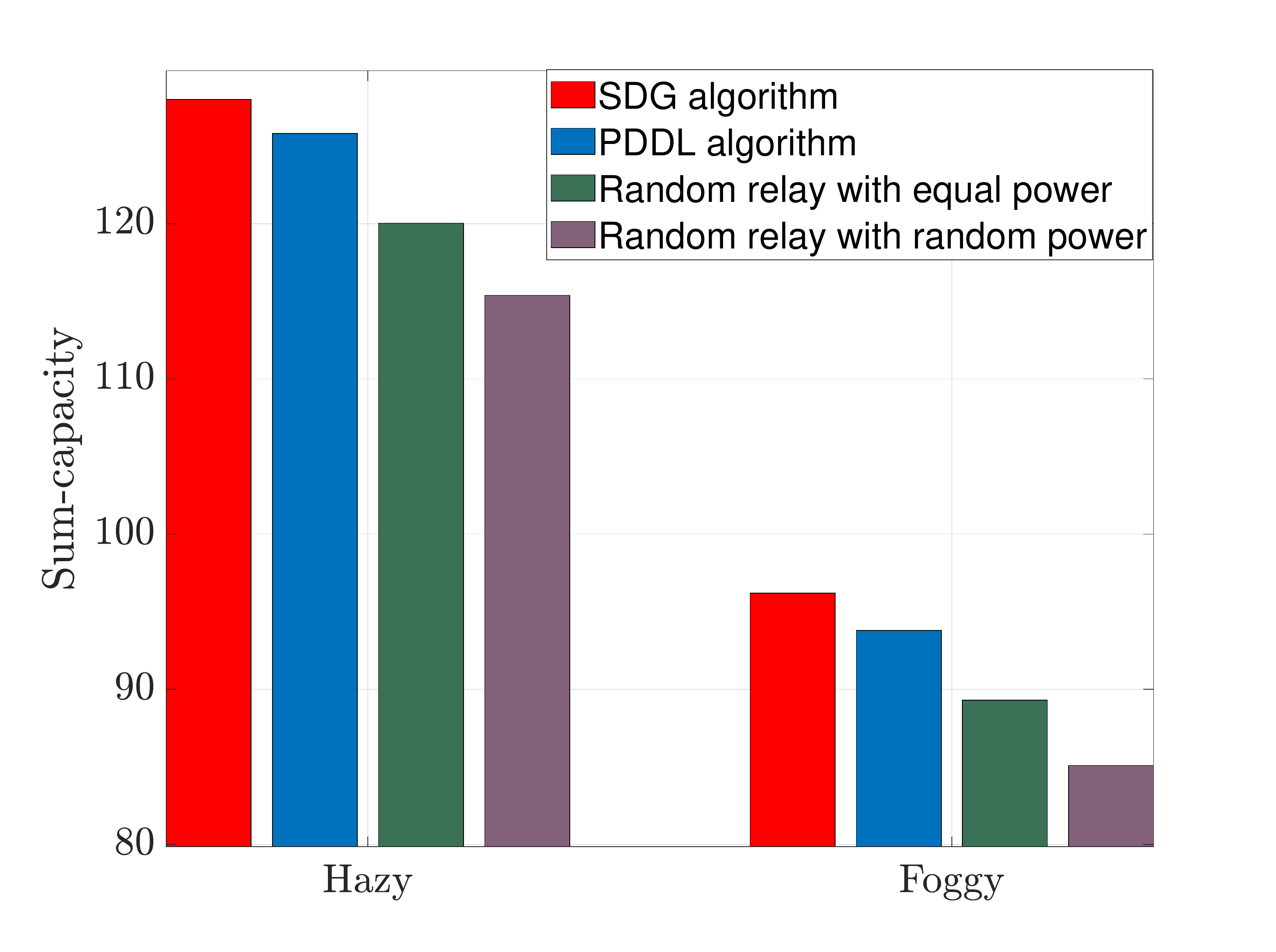}%
\caption{}%
\label{subfigd_pr}%
\end{subfigure}%
\caption{Performance of the SDG, the PDDL and the baseline policies for joint power and relay allocation in the relay-assisted multichannel FSO network. (a) The objective value. (b) The constraint values for the SDG. (c) The constraint values for the PDDL. (d) Hazy and light foggy weather conditions.}\label{fig_pr_1stexample}\vspace{-5mm}
\end{figure*}
%%%%%%%%%%%%%%%%%%%%%%%%%%%%%%%%%%%%%%%%%%%%%%%%%%%%%%%%%%%%%%%%%%%%%%%%%%%%%%%%%%%%%%%%%%%%%%%%%%%%%%%%%%%%%%%%%%%%%%%%%%%%%%%%%%%%%%%%%%%%%%%%%%%%%%%%%%%%%%%%%%%%%%%%%%%%%%%%%%%%%%%%%%%%%%%%%%%%%%%%%%%%%%%%%%%%%%%%%%%%%%%%%%%%%%%%%%%%%%%%%%%%%%%%%%%%%%%%%%%%%%%%%%%%%%%%%%%%%%%%%%%%

We assume a $1$-hop network with $M=5$ parallel relays per hop and $L=5$ orthogonal optical carriers. The priority weights $\bbomega$ are drawn randomly in $[0,1]$ and system parameters are set as: $B=5\times 10^8 Hz$, $T_f = 10^{-8}s$, $\epsilon=1$, $P_t=1.5W$, $P_s = 0.6W$ and $R=0.75 A/W$. We consider two model-free baseline policies: (i) the random relay selection with equal power allocation and (ii) the random relay selection with random power allocation\footnote{The objective is more complicated and the allocated resources include both continuous and binary variables, such that model-based algorithms that solve the KKT conditions (as the water-filling algorithm in Section \ref{exppower}) requires more careful relaxations and much more computation. As we have shown the theoretical (Section \ref{sec_sdg}) and numerical (Section \ref{exppower} and \ref{exp:relay}) optimality of the SDG, we focus on verifying performance of our algorithms in this scenario compared to two low-complexity model-free baseline policies.}. For the PDDL algorithm, we consider the truncated Gaussian distribution for allocated powers and the categorical distribution for selected relays. The DNN is constructed as a two-layered architecture of $200$ and $100$ hidden units with the ReLU nonlinearity. The CSI $\bbh$ are fed as inputs to the DNN, which outputs parameters that specify policy distributions $\pi_{\bbh,\bbtheta}$.

Fig. \ref{fig_pr_1stexample} plots the objective and constraints of the SDG, the PDDL and two baseline policies. The performance of the SDG and the PDDL is superior to that of baseline policies, and the performance improvements get emphasized compared to either single power adaptation in Section \ref{exppower} or single relay selection in Section \ref{exp:relay}. This is because advantages of our algorithms get compounded in this joint problem. The PDDL obtains close performance to the SDG but does not require any system model for implementation. From Fig. \ref{subfigb_pr} and \ref{subfigc_pr}, we see that constraint values converge to zero for both the SDG and the PDDL, confirming the solution feasibility. Fig. \ref{subfigd_pr} shows that both algorithms apply well to changing weather conditions.

\textbf{FSO fronthaul network.} The second experiment considers FSO fronthaul networks---see Section \ref{powerrelay}. We consider the large-scale fronthaul network divided into multiple small-scale fronthaul clusters that perform resource allocation independently. In a fronthaul cluster, the goal is to allocate powers to optical carriers and select the optimal AN at each RRH that maximize the sum-capacity. The optimization problem is formulated by the objective \eqref{eq_fronthaulFSO_capacity} and constraints \eqref{eq_fronthaulFSO_constraints}. Note that the data congestion constraints \eqref{eq_fronthaulFSO_constraintsc} further complicate the problem, making it extremely challenging to solve in practice.

\begin{figure*}%
\centering
\begin{subfigure}{0.33\columnwidth}\centering
\includegraphics[width=1\linewidth, height = 0.7\linewidth]{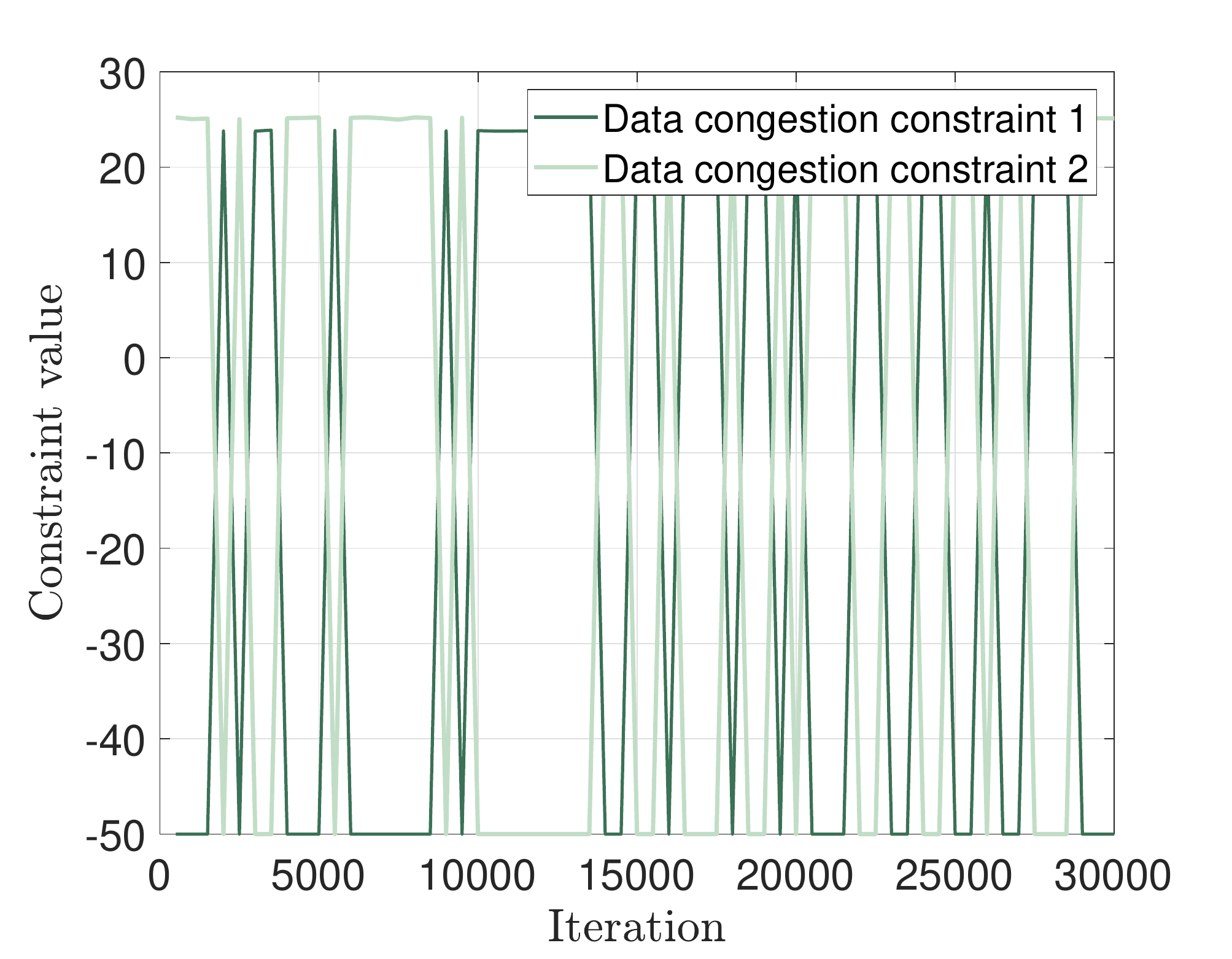}%
\caption{}%
\label{eqconstraint}%
\end{subfigure}\hfill\hfill%
\begin{subfigure}{0.33\columnwidth}\centering
\includegraphics[width=1\linewidth,height = 0.7\linewidth]{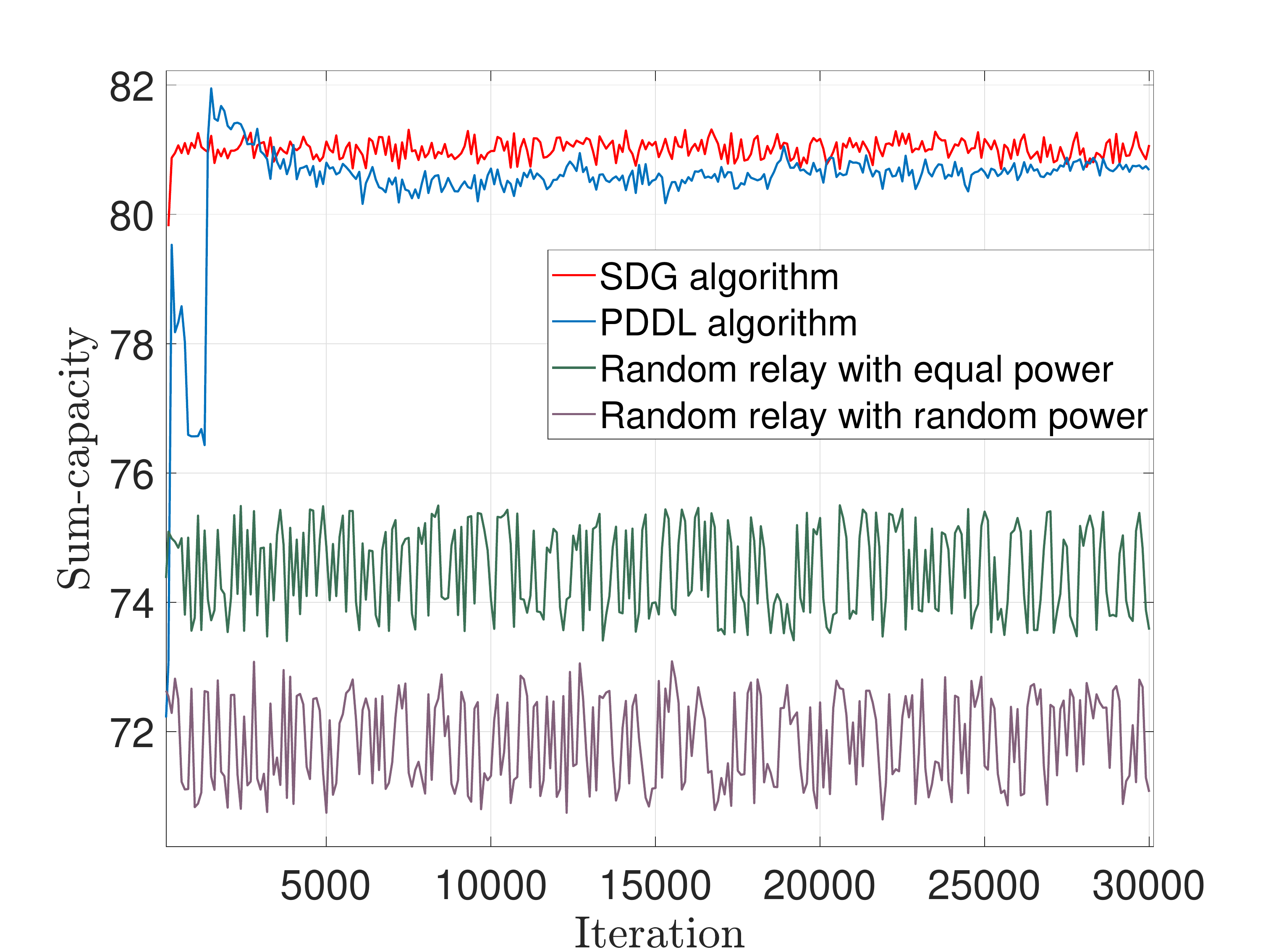}%
\caption{}%
\label{prcom}%
\end{subfigure}\hfill\hfill%
\begin{subfigure}{0.33\columnwidth}\centering
\includegraphics[width=1\linewidth,height = 0.7\linewidth]{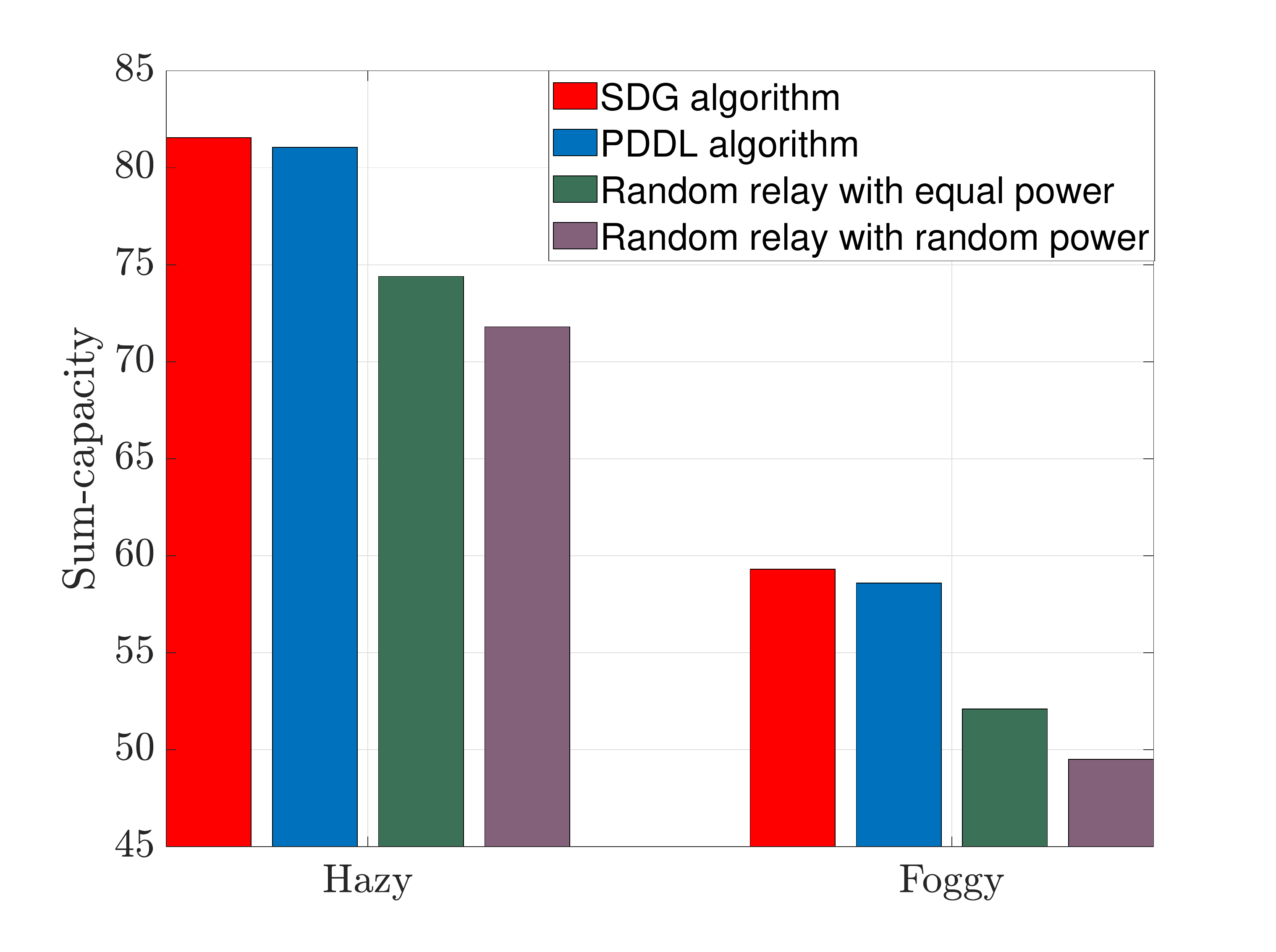}%
\caption{}%
\label{prcomd}%
\end{subfigure}
\caption{(a) The data congestion constraint values of the random AN selection with equal power allocation. (b) The objective values of the SDG, the PDDL and the baseline policies for joint power and relay allocation in the FSO fronthaul network. (c) Hazy and light foggy weather conditions.}\label{fig_pr}\vspace{-5mm}
\end{figure*}

We consider a FSO fronthaul cluster with $N=5$ RRHs, $M=2$ ANs, one BBU and $L=5$ orthogonal carriers. RRHs and ANs are distributed randomly at locations in the squares $[-5km, 5km]^2$ and $[-1km,1km]^2$, respectively. The system parameters are set as: $B=10^9 Hz$, $T_f = 10^{-9}s$, $\epsilon=1$, $P_t=1.5 W$, $P_s=0.6W$, $R=0.75 A/W$ and $C_t = 50$. For the PDDL algorithm, the truncated Gaussian distribution and the categorical distribution are used for allocated powers and selected ANs. As the problem becomes more complicated, we consider a denser DNN with $3$ layers, each of which contains $400$, $200$ and $100$ hidden units, and the ReLU nonlinearity is used. The inputs are channel conditions $\bbh$ over the network and the outputs are parameters specifying policy distributions. We point out that, due to the data congestion constraints \eqref{eq_fronthaulFSO_constraintsc}, there is no feasible heuristic baseline policy for comparison in this problem. To show this more precisely, we test the random relay selection with equal power allocation and plot data congestion constraint values in Fig. \ref{eqconstraint}. We see that constraints are easily broken due to the fact that, if each RRH selects a AN randomly, it is possible that there exist times when all RRHs transmit to the same AN violating the data congestion constraints. We can nonetheless validate the performance of the proposed SDG and PDDL algorithms and demonstrate their feasibility.

\begin{figure*}%
\centering
\begin{subfigure}{0.25\columnwidth}\centering
\includegraphics[width=1.0\linewidth, height = 0.7\linewidth]{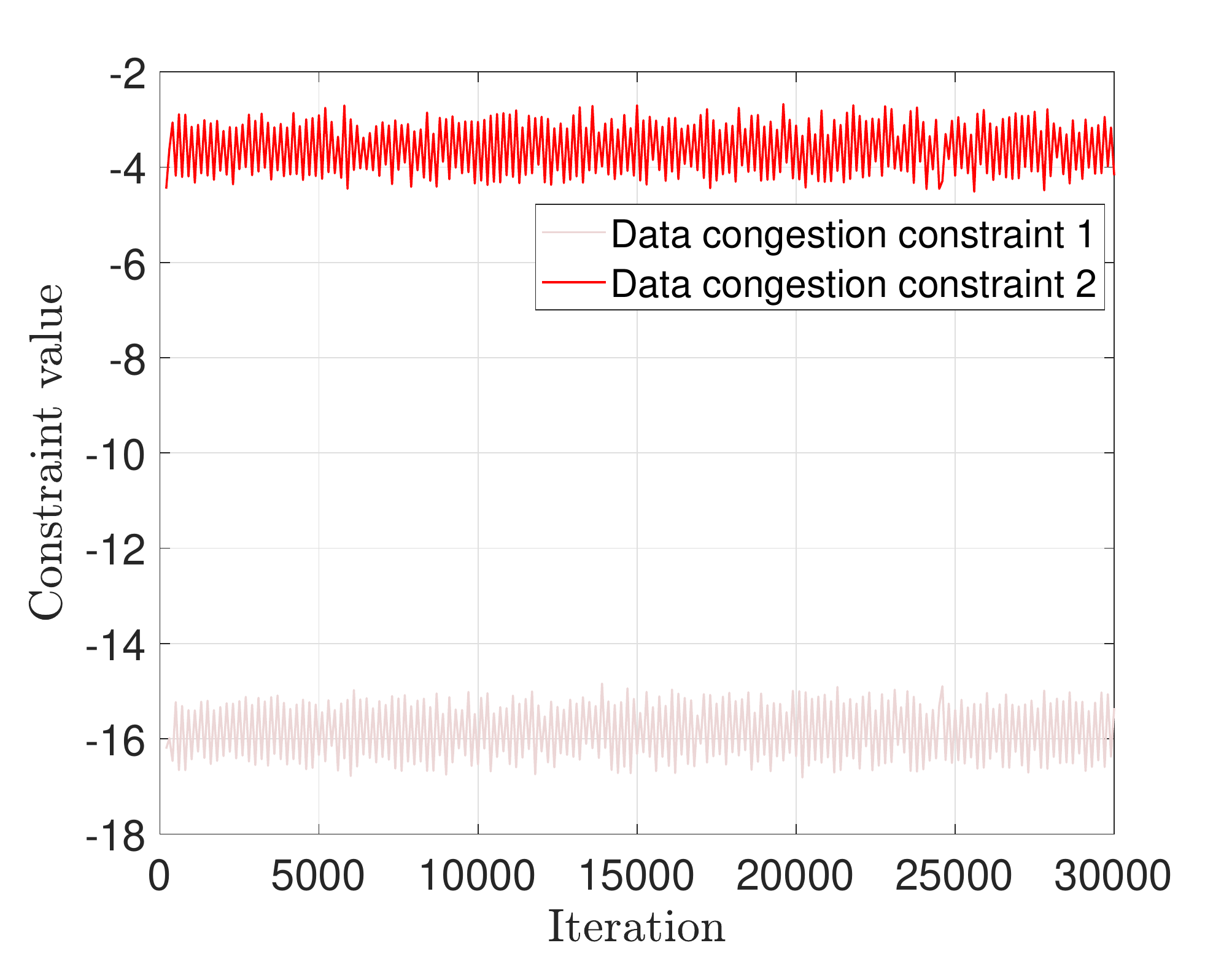}%
\caption{}%
\label{subfig1}%
\end{subfigure}\hfill\hfill%
\begin{subfigure}{0.25\columnwidth}\centering
\includegraphics[width=1.0\linewidth,height = 0.7\linewidth]{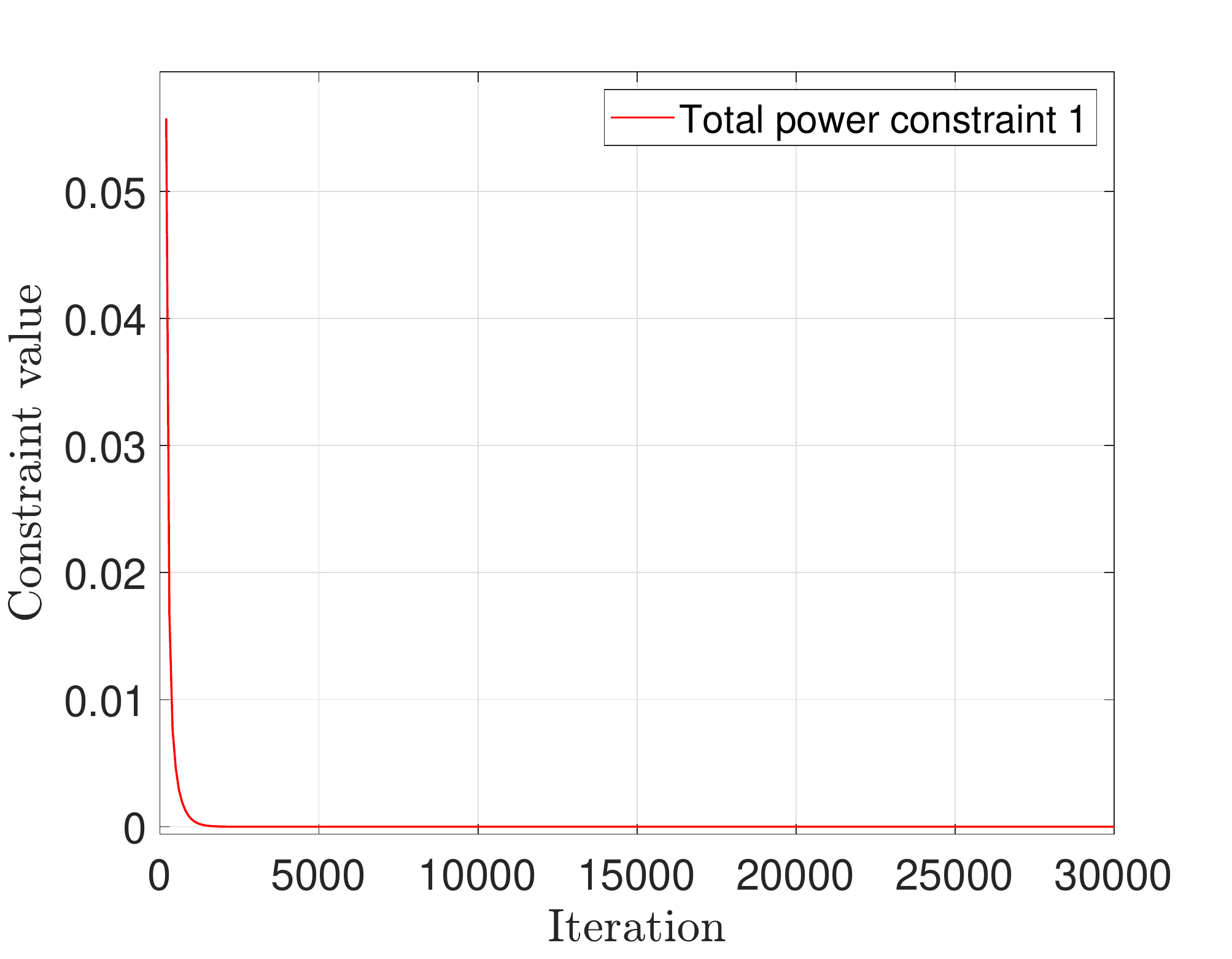}%
\caption{}\label{subfig222}%
\end{subfigure}\hfill\hfill%
\begin{subfigure}{0.25\columnwidth}\centering
\includegraphics[width=1.0\linewidth, height = 0.7\linewidth]{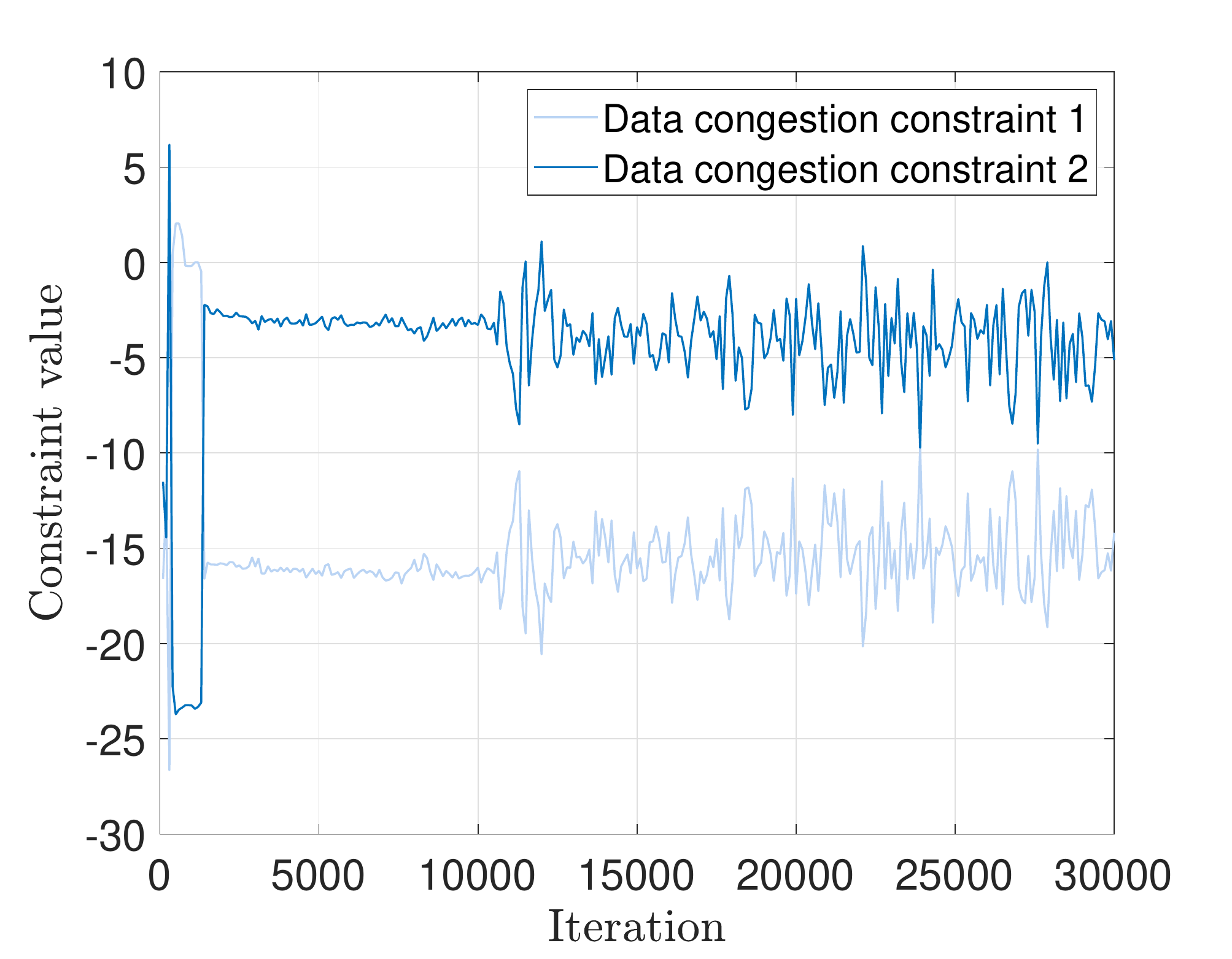}%
\caption{}%
\label{subfig11}%
\end{subfigure}\hfill\hfill%
\begin{subfigure}{0.25\columnwidth}\centering
\includegraphics[width=1.0\linewidth,height = 0.7\linewidth]{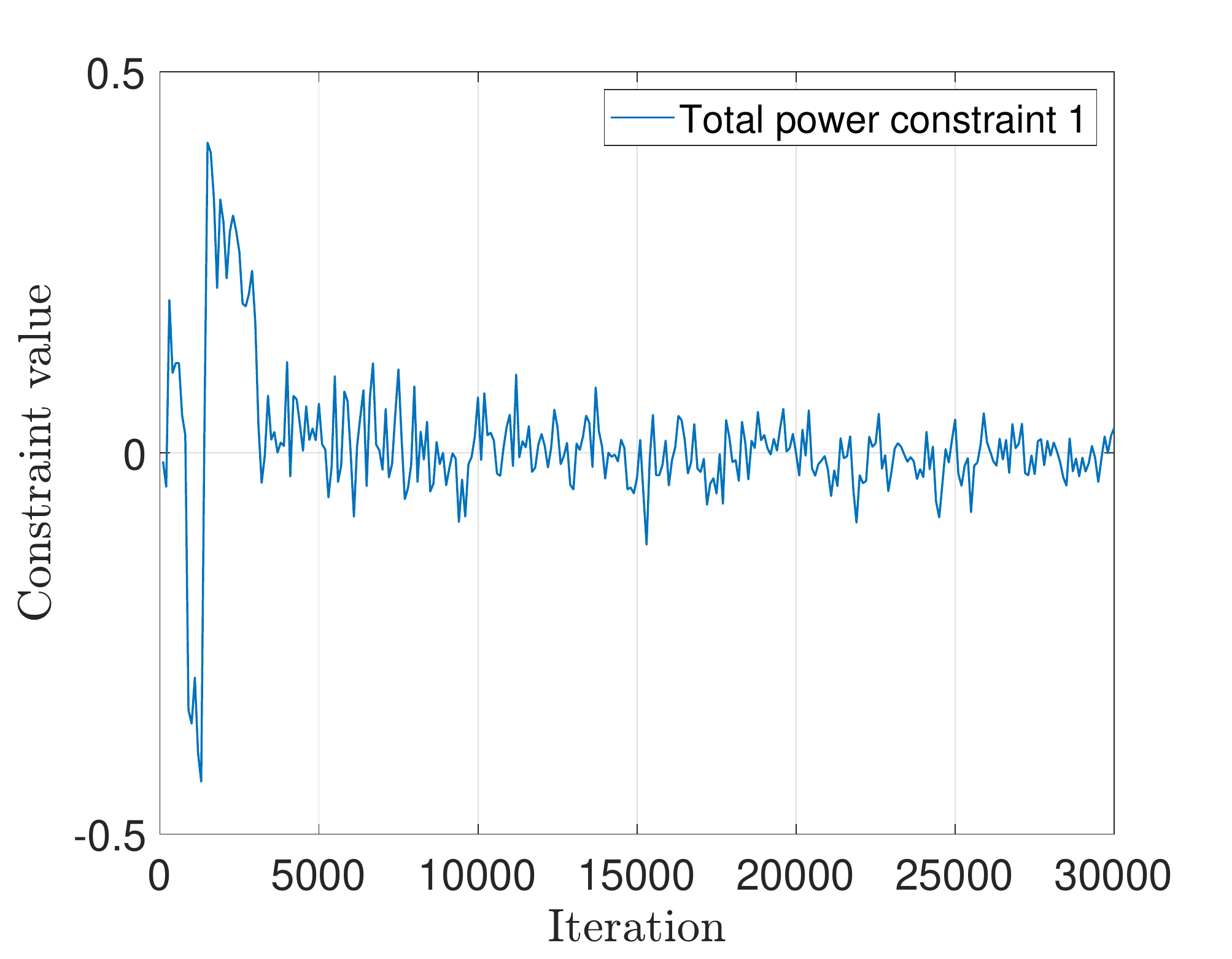}%
\caption{}\label{subfig2222}%
\end{subfigure}
\caption{(a) The data congestion constraints of the SDG (b) The total power constraint of the SDG at an example RRH. (c) The data congestion constraints of the PDDL. (d) The total power constraint of the PDDL at an example RRH.}\label{fig_fronthaul_constraint_values}\vspace{-5mm}
\end{figure*}

%%%%%%%%%%%%%%%%%%%%%%%%%%%%%%%%%%%%%%%%%%%%%%%%%%%%%%%%%%%%%%%%%%%%%%%%%%%%%%%%%%%%%%%%%%%%%%%%%%%%%%%%%%%%%%%%%%%%%%%%%%%%%%%%%%%%%%%%%%%%%%%%%%%%%%%%%%%%%%%%%%%%%%%%%%%%%%%%%%%%%%%%%%%%%%%%%%%%%%%%%%%%%%%%%%%%%%%%%%%%%%%%%%%%%%%%%%%%%%%%%%%%%%%%%%%%%%%%%%%%%%%%%%%%%%%%%%%%%%%%%%%%
\begin{table}[t] 
\begin{center}  
\caption{Implementation time required for the SDG and the PDDL for joint power and relay allocation. (a) Relay-assisted multichannel FSO network. (b) FSO fronthaul network.}  
\label{table3}
\begin{tabular}{|l|l|l| p{2cm}|}  
\hline  
 & Case (a)  & Case (b) \\ \hline  
The SDG &  $8.28\cdot 10^{-3}$s & $2.83\cdot 10^{-2}$s \\ \hline  
The PDDL &  $3.28\cdot 10^{-5}$s & $9.34 \cdot 10^{-5}$s\\
\hline  
\end{tabular}  
\end{center}  \vspace{-4mm}
\end{table} 

We plot in Fig. \ref{prcom} the performance achieved by the SDG and the PDDL algorithms. Since both baseline policies are not feasible, we consider them as benchmark values only for reference. As seen in the prior simulations, the SDG exhibits the best performance using model knowledge, while the PDDL performs comparably to the SDG while forgoing any models. Fig. \ref{fig_fronthaul_constraint_values} plots constraint values for both algorithms to confirm feasibility of the learned solutions, where Fig. \ref{subfig222} and \ref{subfig2222} show total power limitations, while Fig. \ref{subfig1} and \ref{subfig11} illustrate data congestion constraints.

To conclude our numerical analysis, we provide in Table \ref{table3} the implementation time of the SDG and the PDDL algorithms for two joint power and relay allocation problems. We see that besides requiring system model information, the SDG achieves better performance at the expense of more implementation time. The latter gets emphasized when the FSO system or the resource allocation problem becomes more complicated. The implementation time of the PDDL is much lower but increases slightly from single power adaptation to the joint power and relay allocation, which is because the applied DNN gets deeper and denser. However, its computation time is independent on the FSO system and the optimization problem, resulting in an efficient implementation. We further note that a denser DNN learns better performance while taking more time for implementation, indicating a tradeoff between these two factors.

%!TEX root = mainOp.tex
%%%%%%%%%%%%%%%%%%%%%%%%%%%%%%%
%%% SECTION : Conclusions   %%%
%%%%%%%%%%%%%%%%%%%%%%%%%%%%%%%

\section{Conclusions} \label{sec:conclusion}

In this paper, we consider the general resource allocation in free space optical communications. We formulate the problem under the constrained stochastic optimization framework. Such problems are typically challenging due to the non-convex nature, multiple constraints and lack of model information. We first proposed the model-based Stochastic Dual Gradient algorithm, which solves the problem exactly by exploiting the strong duality. However, it heavily relies on system models that may not be available in practice. The model-free Primal-Dual Deep Learning algorithm was developed to overcome this issue. It parameterizes the resource allocation policy with DNNs and learns optimal parameters by updating primal and dual variables simultaneously. Policy gradient method is applied to the primal update in order to estimate necessary gradient information without using the knowledge of system and channel models. The proposed algorithms are computationally efficient and transferable to any resource allocation problem under the framework, which were validated in numerous numerical experiments.

\end{spacing}

%%%%%%%%%%%%%%%%%%
%%% APPENDIX   %%%
%%%%%%%%%%%%%%%%%%

%\appendices %\label{sec:appendix}

%\input{proofsStabilitySO3.tex}

\begin{spacing}{1.225}

\bibliographystyle{IEEEtran}
\bibliography{myIEEEabrv,biblioOp}

%%%%%%%%%%%%%%%%%%%%%%%%
%%% EXTRA APPENDIX   %%%
%%%%%%%%%%%%%%%%%%%%%%%%
\newpage

%\label{sec:appendix}

%\input{proofsStabilitySO31.tex}

\end{spacing}
\end{document}